\newcommand{\CLASSINPUTbottomtextmargin}{.75in}
\newcommand{\CLASSINPUTinnersidemargin}{.7in} 
 \documentclass[conference, 10pt, onecolumn]{IEEEtran}
\IEEEoverridecommandlockouts
\usepackage[utf8]{inputenc} 
\usepackage[table]{xcolor}
\usepackage{graphicx}
\graphicspath{{figs/}}

\usepackage{epstopdf}
\usepackage{pgfplots}
\pgfplotsset{compat=1.14}
\usepackage{tikz}
\usetikzlibrary{arrows, shapes}
\usepackage{tabularx}
\usepackage{subfig}
\usepackage{dcolumn}
\usepackage{multirow}
\usepackage{float}
\usepackage{amsthm}
\usepackage[cmex10]{amsmath}
\usepackage{amssymb}
\usepackage[numbers,sort&compress]{natbib}
\usepackage{braket}
\usepackage{array}
\usepackage{longtable}
\usepackage{arydshln}
\usepackage{bm}
\usepackage{bbm}
\usepackage{bbold}

\usepackage{mathtools}

\newcommand{\Fig}[1]{Fig.~\ref{#1}}
\usepackage[T1]{fontenc}
\usepackage{url}
\usepackage{ifthen}
\usepackage{cite}
\usepackage[cmex10]{amsmath}

\makeatletter
\newcommand{\distas}[1]{\mathbin{\overset{#1}{\kern\z@\sim}}}%
\newsavebox{\mybox}\newsavebox{\mysim}
\newcommand{\distras}[1]{%
  \savebox{\mybox}{\hbox{\kern3pt$\scriptstyle#1$\kern3pt}}%
  \savebox{\mysim}{\hbox{$\sim$}}%
  \mathbin{\overset{#1}{\kern\z@\resizebox{\wd\mybox}{\ht\mysim}{$\sim$}}}%
}

\newtheorem{theorem}{Theorem}[section]
\newtheorem{lemma}[theorem]{Lemma}
\newtheorem{claim}{Claim}
\newtheorem{define}{Definition}

\newtheorem{remark}{Remark}

\newcommand{\Pcc}[2]{P_{c_{[#1][#2]}}}

\newcommand{\norm}[1]{\left\| #1 \right\|}
\newcommand{\Pb}[1]{\mathbb{P}\left[#1 \right]}

\newcommand{\Ex}[1]{\mathbb{E}\left[#1 \right]}

\def\PPA{P^{\perp}_{A_1}}

\def\hcc{\sum_{i\in S_2^*}H_i c_i}

\def\PT{P_{tot}}
\def\PTX{P_{tot,\nu}}
\def\phso{\left(1+P'\sum_{i\in S_0\setminus S_2^*}|H_i|^2\right)}

\def\cn{\mathcal{CN}}

\def\cd{\mathcal{C}}
\def\calC{\mathcal{C}}

\def\Ieee{IEEEeqnarray*}
\def\Ieeen{\IEEEyesnumber}

\usepackage{algorithm}
\usepackage{algpseudocode}
\algnewcommand\algorithmicinput{\textbf{Input:}}
\algnewcommand\INPUT{\item[\algorithmicinput]}
\algnewcommand\algorithmicoutput{\textbf{Output:}}
\algnewcommand\OUTPUT{\item[\algorithmicoutput]}

\begin{document}

\title{
Energy efficient coded random access for the wireless uplink\\
\thanks{This material is based upon work supported by the National Science Foundation under Grant No CCF-17-17842, CAREER award under grant agreement CCF-12-53205, the Hong Kong  Innovation and Technology Fund (ITS/066/17FP) under the HKUST-MIT Research Alliance Consortium and a grant from Skoltech--MIT Joint Next Generation Program (NGP).} 
\thanks{The research of A. Frolov and K. Andreev was supported by the Russian Science Foundation (project no. 18-19-00673).}
}

\author{
  \IEEEauthorblockN{Suhas S Kowshik\\}
  \IEEEauthorblockA{MIT\\
                    suhask@mit.edu\\}
  \and
  \IEEEauthorblockN{Kirill Andreev\\}
  \IEEEauthorblockA{SkolTech (Moscow)\\
                    k.andreev@skoltech.ru\\}
  \and
  \IEEEauthorblockN{Alexey Frolov\\}
  \IEEEauthorblockA{SkolTech (Moscow)\\
                    al.frolov@skoltech.ru\\}
  \and
  \IEEEauthorblockN{Yury Polyanskiy\\}
  \IEEEauthorblockA{MIT\\
                    yp@mit.edu\\}
}

\maketitle
\thispagestyle{plain}
\pagestyle{plain}

\begin{abstract}
We discuss the problem of designing channel access architectures for enabling fast, low-latency, grant-free and
uncoordinated uplink for densely packed wireless nodes. Specifically, we study random-access codes, previously
introduced for the AWGN multiple-access channel (MAC) by Polyanskiy'2017, in the practically more
relevant case of users subject to Rayleigh fading, when channel gains are unknown to the decoder. We propose a random
coding achievability bound, which we analyze both non-asymptotically (at finite blocklength) and asymptotically. As a candidate practical solution, we propose an explicit 
sparse-graph based coding scheme together with an alternating belief-propagation decoder. The
latter's performance is found to be surprisingly close to the finite-blocklength bounds. Our main
findings are twofold. First, just like in the AWGN MAC we see that jointly decoding large number of users leads to a
surprising phase transition effect, where at spectral efficiencies below a critical threshold (5-15 bps/Hz depending on
reliability) a perfect multi-user interference cancellation is possible. Second, while the presence of Rayleigh fading significantly increases the minimal required energy-per-bit $E_b/N_0$ (from about
0-2 dB to about 8-11 dB), the inherent randomization introduced by the channel makes it much easier to attain
the optimal performance via iterative schemes. 

In all, it is hoped that a principled definition of the random-access model together with
our information-theoretic analysis will open the road towards unified benchmarking and comparison performance of various 
random-access solutions, such as the currently discussed candidates (MUSA, SCMA, RSMA) for the 5G/6G.
\end{abstract}

\section{Introduction}
\label{sec:intro}

Presently, wireless networks are starting to see a new type of load (a so-called mMTC or machine-type communication), in which hundreds of thousands of devices are serviced by a single base station, each communicating very small and infrequent data payloads. In the interest of reducing hardware complexity, reducing latency and improving energy consumption, the conceptual paradigm shift is to move to the \textit{grant-free} access management, in which uplink communication is not orthogonalized by the base-station (as is done in today's systems). This requires new kinds of codes that can be decoded from uncoordinated and colliding  transmissions.

In this work, we aim to understand the fundamental tradeoffs of these dense random access systems, and provide coding solutions that are close to achieving these fundamental limits. Specifically, we consider a problem of a large number of nodes (potentially unbounded) with any $K_a$ of them communicating to a single access point or base station (BS) over a frame synchronous multiple access channel (MAC) with frame length $n$.

An information-theoretic formulation of this problem was done in \citep{polyanskiy2017perspective} where the author considered an additive white Gaussian noise (AWGN) random access channel (RAC) model. In this formulation the random access is modeled as follows: each of $K_a$ active users encodes his $k$-bit message into an $n$-symbol codeword. The receiver observes superposition of $K_a$ codewords corrupted by the AWGN.
There are a number of challenges in this model: finite blocklength (FBL) effects due to small payload size, massive number of users (comparable to blocklength), sparsity due to random access and incorporating accurate channel models. However, the most crucial departure from canonical MAC is that the users are required to share the same codebook (i.e. they are unidentifiable, unless they desire to put their identity as part of the $k$-bit payload), and the decoder is only required to provide an unordered list of user messages. In the follow-up works, this problem has also been called \textit{unsourced random access}~\citep{vem2017user,amalladinne2018coupled,fengler2019massive}. Another important aspect of this new formulation is the notion of per-user probability of error (PUPE) which is defined as the average (over the active users) fraction of the transmitted messages that are misdecoded. (Recall that classical definition declares error even if any one of the messages decoded incorrectly.)



In a quest towards low-complexity schemes achieving FBL bounds above, a scheme based on concatenated codes (with an
inner binary linear code and an outer BCH codes) in conjunction with a protocol called T-fold ALOHA was considered in
\citep{ordentlich2017low}. $T$--fold ALOHA is a modification of the standard slotted aloha protocol, in that up to $T$
collisions can be decoded in a slot. So, slotted ALOHA corresponds to $T=1$. The idea of $T$--fold ALOHA itself is not
new as the idea of employing multi-packet receivers to resolve small order collisions has reappeared
periodically~\citep{ghez1988stability} and more recently~\citep[Appendix A]{liva2011graph}. The gap between this low
complexity scheme in \citep{ordentlich2017low} and the FBL bound \citep{polyanskiy2017perspective} was reduced in
\citep{vem2017user} by employing a serial interference cancellation scheme on top of an interleaved LDPC code.
Achievability bounds for serial interference cancellation scheme (also known as irregular repetition slotted
ALOHA~\citep{liva2011graph}) were further improved in \citep{andreev-wcnc19}, where density evolution
method~\citep{liva2011graph} and a finite length random coding bound for the Gaussian
MAC~\citep{polyanskiy2017perspective} were combined. In \citep{10.1007/978-3-030-01168-0_15} the LDPC portion of
\citep{vem2017user} was improved by optimizing the protograph of LDPC code for Gaussian MAC using generalized PEXIT
charts. Further improvements were obtained in \citep{amalladinne2018coupled} by developing a compressive sensing based
algorithm. In~\citep{pradhan2019joint} the idea of sparsifying collisions, inherent in T-fold ALOHA, was modified by
randomizing (sparse) locations of the LDPC codeword symbols and by optimizing degree distributions via a suitable
approximation of a density evolution. 

Finally, we mention that there is another promising idea, proposed in 2001 by
Muller and Caire~\citep{caire2001optimal}, that uses non-orthogonal CDMA spreading coupled with an outer code. The key
idea is to demodulate CDMA by leveraging the soft information from the outer decoder (and alternate between the two).
In~\citep{caire2001optimal} authors observed a perfect multi-user cancellation effect, shown to exist also for the
fundamental limit in~\citep{ZPT-isit19}. It remains to explore whether this method is competitive for practically
relevant blocklengths.

Another set of works considers the problem of sending a (distributedly detected) alarm signal with high-reliability on
top of the regular low-rate update traffic, cf.~\citep{stern2019massive}. 

All of the references above focused on the AWGN RAC (or, equivalently, assumed perfect power control of the users'
transmissions equalizing received powers). In the presence of fading and MIMO, there have been various works on
algorithms for on/off activity detection \citep{chen2018sparse,liu2018sparse,liu2018massive} that use compressive
sensing ideas along with approximate message passing algorithm. (We note that the random-access problem can be seen as
on/off activity detection within a population of $2^k$ users, where $k$ is the message length. However, already a
moderate value of $k=100$ precludes the straightforward usage of activity detection protocols.) In
\citep{haghighatshoar2018improved}, scaling laws were derived for activity detection in a massive MIMO scenario. This
and the ideas from \citep{amalladinne2018coupled} have been used to develop a low complexity coding scheme in
\citep{fengler2019massive}. We also note here that our problem can be understood as a sparse support recovery 
in the compressed sensing literature
\citep{wainwright2009information,reeves2012sampling,reeves2009note,reeves2008sparse}. Theoretical investigations in that
literature predominantly consider iid Gaussian codebooks. In particular,
in \citep{reeves2012sampling}, the authors analyze various estimators like maximum likelihood (ML) and linear
estimators like matched filter (MF) and linear minimum mean squared error (LMMSE) but in an asymptotic setting similar
to a many-user MAC \citep{chen2017capacity,polyanskiy2017perspective,ZPT-isit19,kowshik2019fundamental,kowshik2019quasi}
where the number of active users scales linearly in blocklength.

The structure and main contributions of this paper are:
\begin{itemize}
\item In Section~\ref{sec:sys} we formally define the problem of unsourced frame synchronized single antenna
quasi-static Rayleigh fading RAC under per-user error (PUPE). We assume that the channel realizations are not known to the
receiver or the transmitters.
\item A $T$-fold ALOHA access method from~\citep{ordentlich2017low} is reviewed in Section~\ref{sec:fbl}. There are two
ways we apply $T$-fold ALOHA in this paper. One is to get a random-coding (non-constructive) achievability bounds, this
is done in Appendix~\ref{app:achiev}. Another is to use it as part of the explicit construction, which we do in
Section~\ref{sec:ldpc}.
\item A converse (lower) bound on energy-per-bit required for any random-access codes is developed in
Section~\ref{sec:conv}.
\item The random coding achievability and converse bounds are evaluated in the asymptotic setting in
Section~\ref{sec:asymp}.
\item In Section~\ref{sec:ldpc} we develop a low-complexity iterative multi-user decoding scheme based on LDPC codes \citep{Gallager63ldpc,
tanner1981recursive, richardson2008modern} and a belief propagation decoder on a joint Tanner graph.
\item In Section\ref{sec:res} we numerically compare various bounds in the finite-blocklength setting. It is found that
our practical scheme is rather competitive compared to both our own finite-blocklength bounds and asymptotic benchmarks.
\item Section~\ref{sec:conclude} finishes with some future directions.
\end{itemize}

\subsection{Notation}
Let $\mathbb{N}$ denote the set of natural numbers. For $n\in\mathbb{N}$, let $\mathbb{C}^n$ denote the $n$--dimensional complex Euclidean space. Let $S\subset \mathbb{C}^n$. We denote the projection operator or matrix on to the subspace \textit{spanned} by $S$ as $P_{S}$ and its orthogonal complement as $P^{\perp}_{S}$. For $0\leq p\leq 1$, let $h_2(p)=-p\log_2(p)-(1-p)\log_2(1-p)$ and $h(p)=-p\ln(p)-(1-p)\ln(1-p)$, with $0\ln 0$ defined to be $0$. We denote by $\mathcal{N}(0,1)$ and $\cn(0,1)$ the standard normal and the standard circularly symmetric complex normal distributions, respectively. $\mathbb{P}$ and $\mathbb{E}$ denote probability measure and expectation operator respectively. For $n\in \mathbb{N}$, let $[n]=\{1,2,...,n\}$. Lastly, $\norm{\cdot}$ represents the standard euclidean norm.
\section{System Model}
\label{sec:sys}

We follow the definition of a code from \citep{polyanskiy2017perspective}. Fix an integer $K_a\geq 1$ -- the number of active users. Let $\{P_{Y^n|X^n}=P_{Y^n|X_1^n,X_2^n,...,X_{K_a}^n}:\times_{i=1}^{K_a} \mathcal{X}_i^n\to \mathcal{Y}^n\}_{n=1}^{\infty}$ be a multiple access channel (MAC), which is also permutation invariant: for any permutation $\pi$ on $[K_a]$, the distribution $P_{Y^n|X_1^n,...,X_{K_a}^n}(\cdot|x_1^n,...,x_{K_a}^n)$ coincides with $P_{Y^n|X_1^n,...,X_{K_a}^n}(\cdot|x_{\pi(1)}^n,...,x_{\pi(K_a)}^n)$. We also call this a random access channel (RAC).

\begin{define}[\citep{polyanskiy2017perspective}]
\label{def:1}
An $(M,n,\epsilon)$ random-access code for the $K_a$ user MAC $P_{Y^n|X^n}$ is a pair of (possibly randomized) maps $f:[M]\to \mathcal{X}^n$ (the encoder) and $g:\mathcal{Y}^n\to\binom{[M]}{K_a}$ such that if $W_1,...,W_{K_a}$ are chosen independently and uniformly from $[M]$ and $X_j=f(W_j)$ then the average (per-user) probability of error satisfies
\begin{equation}
\label{eq:def1}
P_{e}=\frac{1}{K_a}\sum_{j=1}^{K_a}\Pb{E_j}\leq \epsilon
\end{equation}
where $E_j\triangleq\{W_j\notin g(Y^n)\}\cup \{W_j= W_i \text{ for some } i\neq j\}$ and $Y^n$ is the channel output.
\end{define}

So, all users use the same codebook, and the receiver outputs a list of $K_a$ codewords. Further, the probability of error is the average fraction of incorrectly decoded codewords. 
In the remainder of the paper we particularly focus on the single antenna quasi-static fading MAC:

\begin{equation}
\label{eq:sys1}
Y^n=\sum_{i=1}^{K_a}H_i X_i^n+Z^n
\end{equation}
where $Z^n\distas{}\cn(0,I_n)$, and $H_i\distas{iid}\cn(0,1)$ are the fading coefficients which are independent of
$\{X^n_i\}$ and $Z^n$. 
Consequently, we require each codeword produced by the encoder $f$ to satisfy a maximum power constraint:
\begin{equation}
\label{eq:pow}
\norm{f(w)}^2\leq nP\,, \qquad \forall w \in [M]\,.
\end{equation}

We emphasize that there can be potentially an unbounded number of users, but only $K_a$ of them are active. If each user has a message of size $k$ and transmits at power $P$ per symbol, then the energy-per-bit is given by $E_b/N_0=\frac{nP}{k}$.

In the rest of the paper we drop the superscript $n$ unless it is unclear.

\section{Random-access via $T$-fold ALOHA}
\label{sec:fbl}
In this section, we discuss our main achievability bound based on $T$--fold ALOHA protocol \citep{ordentlich2017low}. 
The idea is the following. Let $T,n_1 \in \mathbb{N}$ such that $T<K_a$ and $n_1<n$. The time slot or frame of length
$n$ is partitioned into $L=n/n_1$ subframes of length $n_1$. The common codebook is of blocklength $n_1$ and thus may
use a larger power $LP$ per degree of freedom. Each user
chooses a slot to send his message uniformly at random independently of other users. If there are $r$ users placing
their codewords in a particular $n_1$-slot, then the law of observations $Y^{n_1}$ and messages $W_1,\ldots,W_r$ in this
slot is given by
\begin{equation}\label{eq:ch_n1}
	Y^{n-1} = \sum_{i=1}^r H_i f(W_i) + Z^{n_1}\,, \qquad W_i \stackrel{iid}{\sim}\mathrm{Unif}[M]\,.
\end{equation}

Suppose there is a code such that
if there are at most $T$ users transmitting in a given block, then with good reliability decoder can estimate all $\le T$
messages, while if $>T$ users were transmitting then no guarantees on the decoder performance are made. For $T=1$ this
corresponds to the usual ``collision model'' prevalent in the analysis of the ALOHA. (Thus $T>1$ serves to partially
address the more realistic physical layer behavior.)
Intuitively, then, if the average number of users per slot, equal to $K_a/L$, is smaller than $T$, then with good
probability all users will be properly decoded.

More specifically, for a given common codebook $\calC \subset B(n_1, \sqrt{n_1 LP})$ inside an $\mathbb{C}^{n_1}$-ball of
radius $\sqrt{n_1 LP}$ and size $|\calC| = M$ we let $P_{e,\text{genie}}(\calC, r)$ denote the following quantity:
$$ P_{e,\text{genie}}(\calC, r) = {1\over r} \sum_{i=1}^r \Pb{W_i \not\in \mathcal{L}(Y^{n-1},r)}\,,$$
where $\mathcal{L}$ is the decoded list of messages. The subindex ``genie'' denotes
the fact that the decoder is aware of the exact number of users active in a slot. Given this genie side-information we
can show that the $T$-fold ALOHA access scheme then attains the overall PUPE for all of $K_a$ users bounded by
$$ 
\epsilon_{T, \text{genie}}(\calC)\triangleq 1-\sum_{r=1}^{T}(1-P_{e,\text{genie}}(\calC,r))\binom{K_a-1}{r-1}\left(\frac{1}{L}\right)^{r-1}\left(1-
\frac{1}{L}\right)^{K_a-r} + {K_a\over M}\,. $$
To get this bound, we first bound the probability that the $i$-th user's message is in collision: $\Pb{\exists j\neq i:
W_j = W_i} \le {K_a -1 \over M}$. Next, we note that the $i$-th user's slot will have $r-1$ other users with probability
$\binom{K_a-1}{r-1}\left(\frac{1}{L}\right)^{r-1}\left(1-\frac{1}{L}\right)^{K_a-r}$. Note that the resulting bound is
monotonically improving with increasing $T$.

\begin{remark} 
\label{remark:1} 
We will use the genie bound for our random-coding constructions and upper bound $P_{e,\text{genie}}$
via~\eqref{eq:Ka2a} in appendix \ref{app:achiev}. Note that the genie assumption prevents the above from being a true
achievability bound. In the AWGN (non-fading)
setting the number of users can be reliably estimated by simply measuring the total received energy in each $n_1$-long
slot, cf.~\citep{ordentlich2017low}. However, in the presence of fading this detector is a lot less reliable.
Consequently, our genie-based bound strictly speaking is only an optimistic estimate of the performance achievable
within a $T$-fold ALOHA scheme by the best possible component subcode.
\end{remark}

To get the true (genie-free) bounds, we are going to use an explicit (LDPC-based) code inside each $n_1$-slot. Our 
decoder automatically detects the number of users in a slot and estimates the messages. To evaluate the performance we
need to define two parameters corresponding to the $n_1$-code $\calC$. Namely, we define $P_e(\calC, r)$ and
$Q_e(\calC,r)$ as follows. Consider the setting of~\eqref{eq:ch_n1}. Fix some decoder (unaware of the number $r$) which
outputs a variable-length list $\mathcal{L} = \mathcal{L}(Y^{n_1}) \subset [M]$. 
We define
\begin{align*}
	P_e(\calC, r) &= {1\over r} \sum_{i=1}^r \Pb{W_i \not \in \mathcal{L}}\,,\\
	Q_e(\calC,r) &= \Pb{ |\mathcal{L}| > r}\,.
\end{align*}

With this definition we get the following bound on the overall PUPE (for all of $K_a$ users):
\begin{IEEEeqnarray*}{LL} 
\label{eq:Ka2b}
\epsilon_{T}(\calC)\triangleq 1-\sum_{r=1}^{T}(1-P_e(\calC,r))\binom{K_a-1}{r-1}\left(\frac{1}{L}\right)^{r-1}\left(1-
\frac{1}{L}\right)^{K_a-r} + {K_a \over M} + q\,,
\IEEEyesnumber
\end{IEEEeqnarray*}
where
$$ q = L \sum_{r=0}^{K_a} {K_a \choose r} L^{-r} (1-{1\over L})^{K_a-r} Q_e(\calC,r)$$
is an upper bound on  $\Pb{\cup_{j=1}^L F_j}$, where $F_j$ is the event that the $j$-th slot's decoded list has size
strictly bigger than the number $r$ of users active in that slot. Note that if the decoder never outputs a list of size
$>T$ then $Q_e(\calC,r)=0$ for all $r\ge T$. In our simulations, we have $Q_e(\calC,r)\approx 0$ (within accuracy of the
Monte Carlo) for all $r\ge 0$. In other words, our decoder does not ever overestimate the number of active users.

\section{Converse bound}
\label{sec:conv}
In this section we describe a simple converse bound based on results from \citep{yang2013quasi} and the meta-converse from \citep{polyanskiy2010channel}. 

\begin{theorem}
\label{th:con_main1}
Let 
\begin{IEEEeqnarray}{LL}
\label{eq:con14}
L_n=n\log(1+PG)+\sum_{i=1}^{n}\left(1-|\sqrt{PG}Z_i-\sqrt{1+PG}|^2\right)\IEEEeqnarraynumspace\\
S_n=n\log(1+PG)+\sum_{i=1}^n \left(1-\frac{|\sqrt{PG}Z_i-1|^2}{1+PG}\right)
\end{IEEEeqnarray}
where $G=|H|^2$ and $Z_i\distas{iid}\cn(0,1)$. Then for every $n$ and $0<\epsilon<1$, any $(M,n-1,\epsilon)$ code for the quasi-static $K_a$ MAC satisfies
\begin{IEEEeqnarray}{LL}
\label{eq:con15}
\log(M)\leq \log(K_a)+\log\frac{1}{\Pb{L_n\geq n\gamma_n}}
\end{IEEEeqnarray}
where  $\gamma_n$ is the solution of 
\begin{IEEEeqnarray}{LL}
\label{eq:con16}
\Pb{S_n\leq n\gamma_n}=\epsilon.
\end{IEEEeqnarray}
\end{theorem}
\begin{proof}
Notice that the converse bound for the case where full CSI is available at the receiver (and/or transmitter) is a converse for the no-CSI case as well. Further, by symmetry, it is enough to get a lower bound on the probability that a particular user's message is not in the decoded list. Finally, we can assume that the decoder has the knowledge of the codewords of all other users. To formalize, let $Y$ be the received vector and let $L(Y)$ be the list of codewords output by the decoder (we use the list of codewords or messages interchangeably). The size of the list is $|L(Y)|\leq K_a$. Then we have the following implications:

\begin{IEEEeqnarray*}{LL}
\label{eq:con1}
&\frac{1}{K_a}\sum_{t=1}^{K_a}\Pb{X_t\notin L(Y)}\geq 1-\epsilon \\
\iff & \Pb{X_1\notin L(Y)}\geq 1-\epsilon \label{eq:con1a} \IEEEyesnumber\\ 
\impliedby & \Pb{X_1\notin L(Y,H_1)}\geq 1-\epsilon \label{eq:con1b}\IEEEyesnumber \\
\impliedby & \Pb{X_1\notin L(Y,H_{[K_a]},X_{[K_a]\setminus \{1\}})}\geq 1-\epsilon\label{eq:con1c}\IEEEyesnumber
\end{IEEEeqnarray*}
where \eqref{eq:con1b} and \eqref{eq:con1c} represents the case when the decoder has access to the fading realization of user 1 and interference from all other users respectively.

Now, given $H_{[K_a]}$ and $X_{[K_a]\setminus \{1\}}$ at the receiver, the channel is equivalent to
\begin{IEEEeqnarray*}{LL}
\label{eq:con2}
Y_1=H_1 X_1+Z
\end{IEEEeqnarray*}
where $H_1$ and $Z$ are same as before, the decoder outputs a list of messages $\hat{W}=L(Y_1,H_1)$ of size at most
$K_a$ and the probability of error is $\Pb{W_1\notin \hat{W}}$ where $W_1\distas{}unif{[M]}$ is the users message.
Observe that this is similar to the case dealt in \citep{yang2013quasi}, but the decoder is performing list decoding.
Using the meta converse variation for list decoding, e.g.~\citep[Proposition 3]{tan2014fixed}, we can modify the converse bound in \citep{yang2013quasi} that results in replacement of $\log M$ with $\log (M/K_a)$. We note that \citep[Lemma 39]{polyanskiy2010channel_ieee} holds here (this is used in the the converse bound of \citep{yang2013quasi}). Combining theses with implications \eqref{eq:con1a}, \eqref{eq:con1b} and \eqref{eq:con1c} we have the theorem.
\end{proof}
\section{Low-complexity iterative coding scheme}
\label{sec:ldpc}
In this section, we present a low-complexity iterative coding scheme based on LDPC codes, which allows one to decode user messages in a slot. 


Recall that the users utilize the same codebook. Let us denote it by $\mathcal{C}$ and explain how to construct it. We start with a binary $[n, k]$ LDPC codebook and replace each $0$ with  $+\sqrt{P}$ and each $1$ with $-\sqrt{P}$. Let us show the bit-wise MAP decoding rule for the $j$-th bit of the $i$-th user below
\begin{\Ieee}{LLL}
\label{eq:argmax_functional}
\hat{X}_{i,j} =  \arg \max\limits_{{X}_{i,j} \in \pm \sqrt{P}} \mathbb{E} \left[ \sum\limits_{\sim {X}_{i,j}} p_{Y|X}\left(Y \: | \: \sum\limits_{k=1}^{T} H_k X_k \right) \prod\limits_{k = 1}^{T} \mathbb{1}_{X_k \in \mathcal{C}} \right]  \Ieeen\\
\end{\Ieee}
where the expectation is taken over ${H_1, H_2, \ldots, H_{T}}$.
Following \citep{richardson2008modern}, the summation ``$\sim {X}_{i,j}$'' means that we sum over all positions in all user codewords, except ${X}_{i,j}$.

\subsection{Alternating BP-decoder general description}
The decoder aims to recover all the codewords based on the received vector $Y$. The decoder employs a low-complexity iterative belief propagation (BP) decoder that deals with a received soft information presented in a log-likelihood ratio (LLR) form. The decoding system can be represented as a graph (factor graph, \citep{kschischang2001factor}), which is shown in \Fig{fig:fg}.

\begin{figure}[ht]
\centering
\input{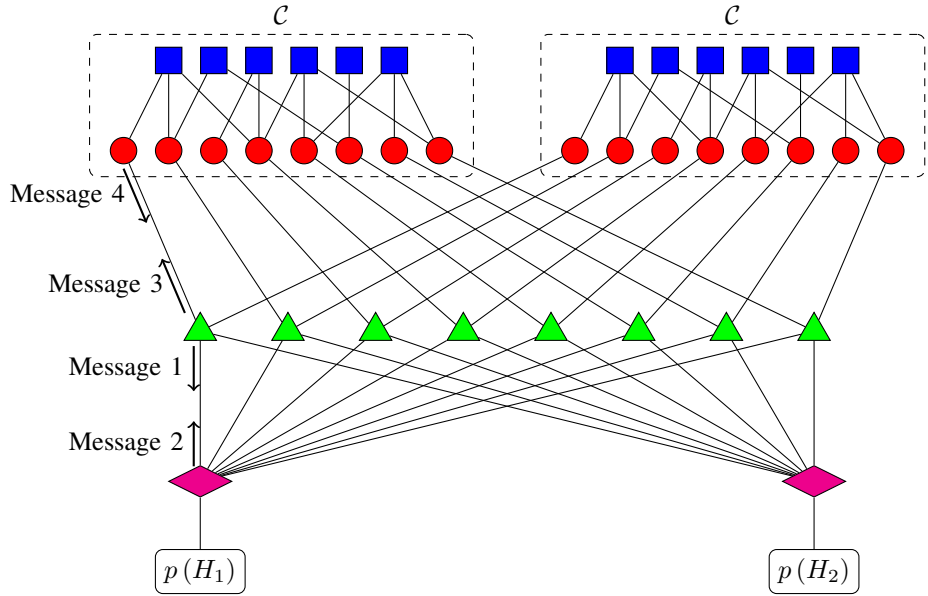}
\def\TotalNodeSpan{(\TannerHorizontalSpacing * (\NVariable - 1) + 2 * \TannerBoxSpacing)}
\def\FuncNodeSpan{0.8 * \TotalNodeSpan}
\def\FuncNodeMargin{0.1 * \TotalNodeSpan}
\def\FuncHorizontalSpacing{\FuncNodeSpan / (\NVariable - 1)}
\def\ArrowInclination{\FuncNodeMargin / \FuncTannerSpacing}
\def\FuncNodeY{-\TannerVerticalSpacing / 2-\FuncTannerSpacing}
\def\FadingNodeY{-\TannerVerticalSpacing / 2 - \FuncTannerSpacing - \FuncFadingSpacing}
\def\FadingPdfNodeY{-\TannerVerticalSpacing / 2 - \FuncTannerSpacing -1.6 *  \FuncFadingSpacing}
\tikzstyle{TannerBox} = [
    rectangle, draw, align=center, rounded corners=1mm, dashed,
    minimum height = {\TannerBoxHeight * (\TannerVerticalSpacing + 2 * \NodeDiameter)}, 
    minimum width =  {\TannerBoxWidth * (\TannerHorizontalSpacing * (\NVariable - 1) + 2 * \NodeDiameter)},
    ]
\tikzstyle{CheckNode} = [rectangle, draw, solid, rounded corners=0mm, fill=blue, minimum height=\NodeDiameter, minimum width=\NodeDiameter]
\tikzstyle{VarNode} = [circle, draw, solid, fill=red, minimum height=\NodeDiameter]
\tikzstyle{FuncNode} = [
regular polygon,
regular polygon sides=3,
draw,
scale=0.75,
minimum height=\NodeDiameter,
fill=green
]
\tikzstyle{FadingNode} = [diamond, aspect=2, draw, solid, rounded corners=0mm, fill=magenta, minimum width=2.4*\NodeDiameter, minimum height=1.2*\NodeDiameter]
\tikzstyle{FadingNodePdf} = [
    rectangle, draw, solid, align=center, rounded corners=1mm,
    minimum height = 1,
    minimum width =  1,
    ]
\tikzstyle{line} = [-,
    draw, solid
]
\def\Tanner {
\foreach \i in {1, ..., \NCheck}
{
    \node at ({(\i - (\NCheck + 1) / 2) * \TannerHorizontalSpacing + \TannerBoxShift}, {\TannerVerticalSpacing / 2}) [CheckNode] (c\i\graphname) {};
}
\foreach \i in {1, ..., \NVariable}
{
    \node at ({(\i - (\NVariable + 1) / 2) * \TannerHorizontalSpacing + \TannerBoxShift}, {-\TannerVerticalSpacing / 2}) [VarNode] (v\i\graphname) {};
}
\path [line] (c1\graphname) -- (v1\graphname);
\path [line] (c1\graphname) -- (v2\graphname);
\path [line] (c1\graphname) -- (v4\graphname);

\path [line] (c2\graphname) -- (v2\graphname);
\path [line] (c2\graphname) -- (v6\graphname);

\path [line] (c3\graphname) -- (v3\graphname);
\path [line] (c3\graphname) -- (v4\graphname);

\path [line] (c4\graphname) -- (v4\graphname);
\path [line] (c4\graphname) -- (v5\graphname);
\path [line] (c4\graphname) -- (v8\graphname);

\path [line] (c5\graphname) -- (v6\graphname);

\path [line] (c6\graphname) -- (v5\graphname);
\path [line] (c6\graphname) -- (v7\graphname);
\path [line] (c6\graphname) -- (v8\graphname);
}
\begin{tikzpicture}[->,
    align=center,
    xscale=\MainScaleX,
    yscale=\MainScaleY,
]
\node at (-\TannerBoxSpacing, 0) [TannerBox] (TannerLeft) {};
\node at (\TannerBoxSpacing, 0) [TannerBox] (TannerLeft) {};
\node [above] at (-\TannerBoxSpacing, 0.8 * \TannerVerticalSpacing) {$\mathcal{C}$};
\node [above] at (+\TannerBoxSpacing, 0.8 * \TannerVerticalSpacing) {$\mathcal{C}$};
{
\def\TannerBoxShift{-\TannerBoxSpacing}
\def\graphname{G1}
\Tanner
}
{
\def\TannerBoxShift{+\TannerBoxSpacing}
\def\graphname{G2}
\Tanner
}
\foreach \i in {1, ..., \NVariable}
{
    \node at ({(\i - (\NVariable + 1) / 2))*\FuncHorizontalSpacing}, \FuncNodeY) [FuncNode] (func\i) {};
}
\foreach \i in {1, ..., \NVariable}
{
    \def\graphname{G1}
    \path [line] (v\i\graphname) -- (func\i);
    \def\graphname{G2}
    \path [line] (v\i\graphname) -- (func\i);
}
\node at ({-\FuncNodeSpan / 2}, \FadingNodeY) [FadingNode] (fading1) {};
\node at ({+\FuncNodeSpan / 2}, \FadingNodeY) [FadingNode] (fading2) {};

\foreach \i in {1, ..., \NVariable}
{
    \path [line] (fading1) -- (func\i);
    \path [line] (fading2) -- (func\i);
}
\node at ({-\FuncNodeSpan / 2}, \FadingPdfNodeY) [FadingNodePdf] (fadingPdf1) {$p\left(H_1\right)$};
\node at ({+\FuncNodeSpan / 2}, \FadingPdfNodeY) [FadingNodePdf] (fadingPdf2) {$p\left(H_2\right)$};
\path [line] (fading1) -- (fadingPdf1);
\path [line] (fading2) -- (fadingPdf2);

\draw [thick]
({-\TotalNodeSpan / 2}, {-\TannerVerticalSpacing / 2 - \MessageArrowMargin * \FuncTannerSpacing})
--
+({\ArrowInclination * \MessageArrowLength * \FuncTannerSpacing}, {-\MessageArrowLength * \FuncTannerSpacing})
node [midway, left={\MessageArrowMargin*\TannerHorizontalSpacing}] (TextNode) {Message 4};

\draw [thick]
({-\FuncNodeSpan / 2- 2 * \MessageArrowMargin *  \FuncTannerSpacing * \ArrowInclination}, {\FuncNodeY +\MessageArrowMargin * \FuncTannerSpacing})
--
+({-\ArrowInclination * \MessageArrowLength * \FuncTannerSpacing}, {\MessageArrowLength * \FuncTannerSpacing})
node [midway, left=0.1*\TannerHorizontalSpacing](TextNode) {Message 3};

\draw [thick]
({-\FuncNodeSpan / 2 - 2 * \MessageArrowMargin * \FuncTannerSpacing*\ArrowInclination*\ArrowInclination}, {\FuncNodeY - \MessageArrowMargin * \FuncFadingSpacing})
--
+(0, {-\MessageArrowLength * \FuncFadingSpacing})
node [midway, left=0.1*\TannerHorizontalSpacing](TextNode) {Message 1};

\draw [thick]
({-\FuncNodeSpan / 2 - 2 * \MessageArrowMargin * \FuncTannerSpacing*\ArrowInclination*\ArrowInclination},{\FadingNodeY+ \MessageArrowMargin * \FuncFadingSpacing})
--
+(0, {\MessageArrowLength * \FuncFadingSpacing})
node [midway, left=0.1*\TannerHorizontalSpacing](TextNode) {Message 2};
\end{tikzpicture}
\caption{Iterative  joint  decoding  algorithm  (alternating  BP-decoder), factor graph}
\label{fig:fg}
\end{figure}
There are four types of nodes in the graph. User LDPC codes are presented with the use of Tanner graphs with variable (red color) and check nodes (blue color). At the same time, there is a third kind of nodes in the figure -- functional nodes (green color). These nodes correspond to the elements of the received vector $Y$. The fourth kind of nodes (magenta nodes) corresponds to fading coefficients. We note that the decoder also performs an estimation of fading coefficients (latent variables).   

The decoding algorithm is based on the iterative message passing procedure. There are two types of iterations in our system: inner iterations, which are used for LDPC code decoding and outer iterations used for fading coefficients estimation. In what follows we mean an outer iteration in all the cases where the type of iteration is not specified. The user codewords are decoded in a sequential manner. Let us consider a single user decoding. This process consists of the calculation and passing of four message types (see \Fig{fig:fg}). We note that both fading coefficients and LLRs for other users remain fixed during this process. Every message is described in details below:

\paragraph{Message type $1$ (from functional nodes to fading nodes)} Without loss of generality let us consider the first functional node. Assume we received a symbol $y$. By $x_i = X_{i,1} \in \{ +\sqrt{P}, -\sqrt{P}\}$, $i = 1, \ldots, T$, we denote symbols sent by the users. Let us show how to calculate a posterior probability density function (pdf) of $H_1$ from the first functional node. We denote this message by $R_1^{(1)}$ and calculate it as follows
\begin{\Ieee}{LLL}
\label{eq:update_h_R}
R_1^{(1)}(h_1) \propto {\mathbb{E}\left[\sum\limits_{x_1, x_2, \ldots x_T} p(y |\sum\limits_{j=1}^{T} H_j x_j)\prod\limits_{j=2}^{T}  \Pr(x_j)\right]}, \Ieeen
\end{\Ieee}
where the expectations are taken over $H_2, \ldots, H_{T}$. Such updates are calculated at every functional node and denoted by $R_1^{(i)}$, $i=1, \ldots, n$.

\paragraph{Message type $2$ (from fading nodes to functional nodes)} We denote the message from $j$-th fading node to $i$-th functional node by $Q_j^{(i)}$, this message is a pdf. To find it we need to calculate the product of incoming messages. Let us consider a message from the first fading to the first functional node, we have
\begin{equation}
\label{eq:update_h_Q}
Q_1^{(1)}(h_1) = \prod\limits_{i=1}^{n} R_1^{(i)}(h_1),
\end{equation}

\begin{remark}
In a conventional message passing algorithm, the outgoing message is calculated based on messages which come through all the edges except its own edge. But here to reduce the complexity we approximate the complicated message update at fading nodes via the product of a few randomly selected incoming messages.
\end{remark}

\paragraph{Message type $3$ (from functional nodes to LDPC codes)}
Let us note, that a posterior LLR for $x_1$ can be calculated as follows
\begin{equation}
\label{eq:decode_func_nodes}
L(x_1) = \log \frac{
\mathbb{E}\left[
    \sum\limits_{x_1 = +\sqrt{P}, x_2, \ldots x_T} p(y | \sum\limits_{j=1}^{T} H_j x_j)\prod\limits_{j=2}^{T} \Pr(x_j)
\right]}
{
\mathbb{E} \left[
    \sum\limits_{x_1 = -\sqrt{P}, x_2, \ldots x_T} p(y | \sum\limits_{j=1}^{T} H_j x_j)\prod\limits_{j=2}^{T} \Pr(x_j)
\right]},
\end{equation}
where the expectations are taken over ${H_1, H_2, \ldots, H_{T}}$ and 
$p(y|a) = \frac{1}{{\pi}} \exp(-(y-a)^2)$. Note that for practical implementation the Monte-Carlo sampling method can be used for expectations.

\paragraph{Message type 4 (LDPC decoding)} After functional nodes decoding one needs to update the LLR for a given user with LDPC iterative decoder. Each user utilizes a standard BP decoding algorithm (Sum-Product or Min-Sum, \citep{richardson2008modern}) to decode an LDPC code.

Now, let us present the final message passing decoding algorithm (see Algorithm~\ref{alg:joint_dec}).

\begin{algorithm}
\caption{Iterative decoding algorithm (alternating BP-decoder)}\label{alg:joint_dec}
\begin{algorithmic}[1]

\State initialize the LLR values of variable nodes for each user code with zero values assuming equal probability for $\sqrt{P}$ and $-\sqrt{P}$ values
\State initialize pdf of $H_i$, $i=1, \ldots, T$. For each coefficient we have pdf for both real and imaginary parts with prior distribution  $\mathcal{N}(0, 1/2)$ corresponding to Rayleigh fading.
\For {$i_O = 1, \ldots, I_{O}$} \Comment perform $I_{O}$ outer iterations
\For {$u = 1, \ldots, T$} \Comment decode users sequentially
\State Propagate message type 1, eq.~\eqref{eq:update_h_R} \Comment from functional nodes to fading nodes
\State Propagate message type 2, eq.~\eqref{eq:update_h_Q} \Comment from fading nodes to functional nodes
\State Sample fading coefficients for further expectation estimation at~\eqref{eq:decode_func_nodes} from the fading coefficients pdfs
\State Propagate message type 3 using sampled fading coefficients, eq.~\eqref{eq:decode_func_nodes} \Comment from functional nodes to LDPC codes
\State Propagate message type 4 \Comment perform $I_{I}$ inner iterations of BP decoder for $u$-th user LDPC code.
\EndFor
\EndFor
\end{algorithmic}
\end{algorithm}
Below the efficient implementation with Gaussian mixtures (GM) approximating the pdf of fading coefficients is discussed.
\subsection{Alternating BP-decoder implementation with Gaussian mixtures}
\label{seq:gm-decoder}
Alternating BP-decoder is based on a successive update of LLRs for every codeword and a successive update of the pdfs of fading coefficients $H_j$, $j=1, 2, \ldots, T$. To construct a practical implementation of this algorithm, one needs a tractable representation of probability density functions that
\begin{itemize}
    \item can be easily manipulated during convolution and multiplication procedures,
    \item retain their form after such kind of transformations through multiple iterations.
\end{itemize}
The simplest form of pdf approximation that satisfies the listed above requirements is the GM model:
\begin{equation}
\label{eq:gm}
\pi\left(\cdot\right)=\sum_{l=1}^{\nu}{\omega_l\mathcal{N}\left(\mu_l, \sigma_l^2\right)}, \quad \sum_{l=1}^{\nu} \omega_l = 1,
\end{equation}
where $\pi$ is the pdf approximation and $\mathcal{N}\left(\mu, \sigma^2\right)$ is the Gaussian pdf with the mean $\mu$ and variance $\sigma^2$. The sum of two random variables having the pdf in the form of~(\ref{eq:gm}) remains a GM. We also note that GM is a conjugate prior with respect to itself, which helps to construct the pdf of $H_j$ given the pdf of~(\ref{eq:update_h_R}) at each functional node.

Now let us specify how the Algorithm~\ref{alg:joint_dec} can be implemented with the use of GMs and describe the steps of every outer iteration. Without loss of generality suppose that the user~$1$ is being decoded.

The first step of outer iteration is to update the fading coefficient for a given user at every functional node (see eq.~(\ref{eq:update_h_R})). This can be done as follows. Rewrite eq.~(\ref{eq:update_h_R}) via GMs. Let us consider the $i$-th functional node. Note that
\begin{equation}
\label{eq:update_h_R_expansion}
H_1^{(i)}x_1^{i} = y_i - \left(\sum\limits_{j=2}^{T}H_j^{(i)} x_j^{(i)} + z\right).
\end{equation}

Given the LLR of some bit $x_j^{(i)}$ and the GM representing the coefficient
$$
H_j^{(i)}\sim \sum_{l=1}^{\nu}{\omega_l\mathcal{N}\left(\mu_l, \sigma_l^2\right)},
$$
the variable $H_j^{(i)} x_j^{(i)}$ will also be a GM in the following form
\begin{equation}
\label{eq:make_llr_prod}
H_j^{(i)} x_j^{(i)} \sim \sum_{l=1}^{\nu}{\omega_l P\left(x_j^{(i)} = +\sqrt{P}\right)\mathcal{N}\left(\sqrt{P}\mu_l, P\sigma_l^2\right)}
 + \sum_{l=1}^{\nu}{\omega_lP\left(x_j^{(i)} = -\sqrt{P}\right)\mathcal{N}\left(-\sqrt{P}\mu_l, P\sigma_l^2\right)}
\end{equation}

As soon as the sum of random variables has the pdf that equals to the convolution of every single pdf, the right-hand side of equation~(\ref{eq:update_h_R_expansion}) is a convolution of GMs. This procedure is straightforward, but the resulting GM component count grows as a product of component counts of every GM under the convolution. One can see, that the $y_i - \sum\limits_{j=2}^{T}H_j x_j^{(i)}$ is also a GM as $y_i$ is a constant. Also, note that this procedure is performed separately for both real and imaginary parts of the signal.

The final step is to construct the $H_1^{(i)}$ pdf given the GM on the right-hand side of eq.~(\ref{eq:update_h_R_expansion}) and the LLR for $x_1^{(i)}$. The coefficient $H_1^{(i)}$ has a GM pdf that can be calculated in exactly the same manner as in equation~(\ref{eq:make_llr_prod}). Suppose the RHS of ~\eqref{eq:update_h_R_expansion} has a pdf
$$
H_1^{(i)}x_1^{i} \sim \sum_{l=1}^{\nu}{\omega_l\mathcal{N}\left(\mu_l, \sigma_l^2\right)}.
$$
then the pdf of the coefficient $H_1^{(i)}$ can be calculated as follows
\begin{equation}
\label{eq:update_h_R_final}
H_1^{(i)}\sim \sum_{l=1}^{\nu}{
\omega_lP\left(x_1^{(i)} = +\sqrt{P}\right)\mathcal{N}\left(\frac{\mu_l}{\sqrt{P}}, \frac{\sigma_l^2}{P}\right)
}
+ \sum_{l=1}^{\nu}{
\omega_lP\left(x_1^{(i)} = -\sqrt{P}\right)\mathcal{N}\left(-\frac{\mu_l}{\sqrt{P}}, \frac{\sigma_l^2}{P}\right)
}.
\end{equation}

The second step of the outer iteration is to derive the fading coefficient estimate $H_1$ given the messages $H_1^{(i)}$ from every functional node~\eqref{eq:update_h_R_final}. This can be done by multiplying the corresponding pdfs (see eq.~(\ref{eq:update_h_Q})). Note that as in the case of convolution, the product of two GMs is also a GM with the number of components equal to the product of the number of components in the multipliers.

The next two steps in the outer iteration are sampling from GM and functional nodes decoding procedure (see eq.~\eqref{eq:decode_func_nodes}).

As it was mentioned before, the final step of the outer iteration is a simple iterative decoding algorithm, that just updates the user's codeword LLRs. The outer iterations are performed over every user successively until the maximum iteration count per user is reached.

\subsection{Gaussian mixture pruning and components merging}
The decoding algorithm performs the convolution and multiplication of multiple GMs at every iteration. In this subsection, these procedures are described in more detail as well as the approach to limiting the ever-growing number of components in the final GM is presented.

The convolution of $GM_{1}\otimes GM_{2}$ with 
$$
GM_{1} = \sum_{l_1=1}^{\nu_1}{\omega_{l_1}\mathcal{N}\left(\mu_{l_1}, \sigma_{l_1}^2\right)}, \quad
GM_{2} = \sum_{l_2=1}^{\nu_2}{\omega_{l_2}\mathcal{N}\left(\mu_{l_2}, \sigma_{l_2}^2\right)}.
$$
is the GM that has $\nu_1\times \nu_2$ components:
\begin{\Ieee}{LLL}
\label{eq:gm_convolution}
GM_{1}\otimes GM_{2} = \sum_{l_1=1}^{\nu_1}{\sum_{l_2=1}^{\nu_2}{\omega_{l_1}\omega_{l_2}\mathcal{N}\left(\mu_{l_1}+\mu_{l_2}, \sigma_{l_1}^2+\sigma_{l_2}^2\right)}}.
\end{\Ieee}
The GM product is given by \eqref{eq:gm_product}. The result has also $\nu_1\times\nu_2$ components.
 
\begin{equation}
\label{eq:gm_product}
GM_{1}\times GM_{2} = \sum_{l_1=1}^{\nu_1}\sum_{l_2=1}^{\nu_2}
\frac{\omega_{l_1}\omega_{l_2}}{\sqrt{2\pi \left(\sigma_{l_1}^2+\sigma_{l_2}^2\right)}}\exp{\left\{-\frac{\left(\mu_{l_1}-\mu_{l_2}\right)^2}{2\left(\sigma_{l_1}^2+\sigma_{l_2}^2\right)}\right\}}
\mathcal{N}\left(
\frac{\sigma_{l_1}^2\sigma_{l_2}^2}{\sigma_{l_1}^2+\sigma_{l_2}^2}\left(
\frac{\mu_{l_1}}{\sigma_{l_1}^2}+\frac{\mu_{l_2}}{\sigma_{l_2}^2}
\right),
\frac{\sigma_{l_1}^2\sigma_{l_2}^2}{\sigma_{l_1}^2+\sigma_{l_2}^2}
\right)
\end{equation}
 
Note, that in practical implementation it is better to manipulate with the logarithm of the Gaussian component weight for numerical stability. The component count optimization procedures are described below and include merge and prune steps.
\subsubsection{Gaussian mixtures pruning}
One can see that both GM convolution and product significantly increase the number of components. For practical implementation, one needs to limit the number of Gaussian components. The first step consists of removing the components with low weights (pruning). This can be easily done by sorting the weights in ascent order and removing several first components whose cumulative weight is less than some threshold.
\subsubsection{Gaussian mixtures components merge}
The GM components which are ``close'' to each other (with the distance measure specified below) must be merged. This approach is described in details in \citep{Clark2006}. The procedure starts from the ``heaviest'' component. All other components that have the distance less than some threshold form a merge-list. This distance can be calculated as follows
$$
d = \frac{\left(\mu_1 - \mu_2\right)^2}{\sigma_1^2}\leq d_{min},
$$
where component 1 has a higher weight than component 2. After the merge list of length $\nu_0$ has been constructed, all the components from this list are replaced by a new component $\omega\mathcal{N}\left(\mu, \sigma^2\right)$ with the following parameters:
$$
\omega = \sum_{l=1}^{\nu_0}\omega_l, \quad
\mu = \sum_{l=1}^{\nu_0}c_l\mu_l,\quad
\sigma^2 = \frac{1}{\omega}\sum\limits_{l=1}^{\nu_0}\mu_l\left(\sigma_l^2 + \left(\mu_l - \mu\right)^2\right)
$$
Note that each component can be merged with any other only once during the GM-merge procedure.

The final step of GM pruning is to apply a hard limit on the maximum components count. This is done for performance stability and helps to control the maximum GM length.

\subsection{Blind detection and error floor}
\label{sec:blind_decoder}
As soon as the iterative decoder operates as an optimization task and this optimization procedure is split between two groups of variables (user LLRs and fading coefficients), one can expect this algorithm to converge to some local maximum of~\eqref{eq:argmax_functional}. Convergence to a local maximum can be a source of the error floor. To overcome the error floor problem one can start the decoding algorithm multiple times and handle functional nodes in random order at every decoding iteration. As soon as GMs are merged and pruned, this provides some source of randomness and pushes the decoding procedure to possibly different local maximums. This approach has eliminated the error floor problem and allowed another opportunity -- a blind detection. Given the multiple decoding attempts, one can select a set of unique codewords that were successfully decoded. Every attempt can detect different codewords. The final output of the decoder is the union of such sets. Without loss of generality, this approach can be applied to the case of unknown user count. As further numerical experiments (see appendix~\ref{app:slot}) show, this approach is a promising one.

\begin{remark}
Even though the number of users in a slot is unknown we never faced with a false alarm problem in our simulations. By false alarm, we mean a situation in which the output list contains codewords that were not transmitted. To explain this fact we note that LDPC codes have a large area of inputs for which they report a failure (the decoder cannot converge to a codeword). Thus, we mention once again that $Q_e(\calC,r)\approx 0$ (within the accuracy of the Monte Carlo) for all $r\ge 0$.
\end{remark}

\begin{remark}
The approach presented in this paper is similar to the approach from \citep{Caire2001}. Nevertheless, the main differences are: a) we consider same codebook case and changed the parallel schedule with serial schedule in order to break symmetry, b) we show that this approach allows to efficiently perform blind user decoding, i.e. determine the number of active users in a slot and recover their messages, c) we suggest an approach how to deal with the error floor caused by the inaccuracy in the estimation of fading coefficients ($H_i$, $i=1, \ldots, T$).
\end{remark}
\section{Numerical results and discussion}
\label{sec:res}
In this section we present the plots of the minimum energy per bit required to achieve a probability of error $\epsilon=0.1$ as a function of $K_a$ for the channel \eqref{eq:sys1}. Figure \ref{fig:1} shows plots of various schemes. The parameters used for evaluation are blocklength $n=30000$ and message size $k=100$ bits. Next we describe how each of these curves was obtained.

For $T$-fold ALOHA using FBL bound, we use the bound for $p_t$ given in \eqref{eq:Ka2a}. For each $K_a$ we find the optimum $L$ (as an optimization over both $L$ and $P$) so that we minimize $E_b/N_0$ such that  the probability of error in \eqref{eq:Ka2b} is less than $0.1$. Since directly optimizing the bound is not easy, we approximate PUPE for the fading channel as \citep{polyanskiy2010channel_ieee}
\begin{\Ieee}{LLL}
\label{eq:Ka_N1}
P_e(M,n_1,r,LP)\approx \\
 \Ex{\mathcal{Q}\left( \frac{n_1 C_{AWGN}(LP\sum_{i=1}^{r}|H_i|^2)-\log_2 M}{\sqrt{n V_{AWGN}(LP\sum_{i=1}^{r}|H_{i}|^2)}}  \right)}\Ieeen
\end{\Ieee}
where $C_{AWGN}(x)=\log(1+x)$ and $V_{AWGN}(x)=1-\frac{1}{(1+x)^2}$ are the capacity and dispersion of a (complex) AWGN channel, respectively. We choose $L$ by using \eqref{eq:Ka_N1} in \eqref{eq:Ka2b}. Then we use the spherical codebook, i.e. codewords uniformly and independently sampled from the (complex) power shell in dimension $n_1=\lfloor n/L\rfloor$ to compute the probability of error according to \eqref{eq:Ka2b} where $P_e(M,n_1,r,LP)$ is computed using brute-force Monte-Carlo simulation of \eqref{eq:Ka2a} with the choice $K_1=K_2=r$. Since $r\leq T$ is small it would not make sense to drop a user. To this end, we produce $2000$ samples, from which we construct the kernel density approximation of the cumulative distributive function (CDF) of the statistic $\max_{\substack{S_0\subset [r]\\ |S_0|=t}}G(Y,S_0,c_{S_0},t)$ (given in \eqref{eq:Ka3aa})  for each $t\leq r$. Then this smooth approximation is used to optimize over $\delta$ in \eqref{eq:Ka2b}. 

For $T$-fold ALOHA using the iterative coding scheme, we have used $(n_1,k)$ LDPC codes with $k=100$ and blocklength $n_1\in\{200,400\}$. We note, that two codes are enough to cover the interval $1 \leq K_a \leq 250$. 
For each of these codes, we get PUPE vs $E_b/N_0$ curves and choose the best code (the best code requires the smallest $E_b/N_0$ in order to achieve $\text{PUPE} \leq \epsilon = 0.1$) for each value of $K_a$. The best waterfall curves for the different number of users are presented in \Fig{fig:2}. Iterative decoder used the multiple component Gaussian mixture model with parameters listed in Table~\ref{tab:sim_settings_gm}. Note again, that in LDPC-based scheme we perform honest blind slot decoding (without assuming the knowledge of user count in a slot).

It can be seen from \Fig{fig:1} that the performance of $T$--fold ALOHA for iterative decoding scheme is very close to that of $T$--fold ALOHA with random coding bounds for small $K_a$. The gap increases with $K_a$ because of our limited choices of LDPC codes, i.e. due to BPSK modulation, we are constrained by $n_1\geq k$. We refer to remark~\ref{remark:1} again to emphasize that the $T$-fold ALOHA with the FBL bound is not a true achievability bound since it assumes that the decoder has knowledge of the number of users in each slot or subframe.

We have also plotted the result of treat interference as noise (TIN) decoding. Here we have used optimistic capacity approximation for PUPE.
\begin{\Ieee}{LLL}
\label{eq:Ka_N2}
\epsilon\approx \Ex{\mathcal{Q}\left( \frac{n C_{AWGN}\left( \frac{  P|H_1|^2}{1+P\sum_{i=2}^{T}|H_{i}|^2}\right)-k}{\sqrt{n V_{AWGN}\left( \frac{  P|H_1|^2}{1+P\sum_{i=2}^{T}|H_{i}|^2}\right)}}  \right)}\Ieeen
\end{\Ieee}

It is easy to get an actual random coding bound for TIN similar to theorem \ref{th:Ka1}, but we don't expect it to be better than \eqref{eq:Ka_N2}.



Also plotted for reference is the Shamai-Bettesh capacity bound from \citep{bettesh2000outages}. It is an asymptotic bound ($n\to \infty$) for the probability of error per-user in the case of symmetric rate and large $K_a$. But, it doesn't assume same codebook. The idea is the following. The joint decoder knows the realization of fading coefficients and users are ranked according the strength of their fading coefficients. It first tries to decode all users. If it fails (i.e., the rate vector is not inside the instantaneous full capacity region), it drops the user with least fading coefficient and tries to decode the remaining $K_a-1$ users. The dropped user forms part of the noise. This process continues iteratively, and the fraction of users that were not decoded is precisely the outage/probability of error per-user. Since the case under discussion is for large $K_a$, the order statistics of the absolute value of fading coefficients crystallize (i.e., become almost non-random) and hence analytical expressions can be derived for outage in terms of spectral efficiency ($kK_a/n$) and total power. So for each $K_a$, we know our operating spectral efficiency and total power, and hence we can use the asymptotic bound to find the probability of error. Most importantly observe that even at $K_a=100$, the random coding based $4$--fold ALOHA performance is off from the capacity bound of \citep{bettesh2000outages} by just $3$~dB. 

The converse from \eqref{eq:con15} and \eqref{eq:con16} is also plotted. This is in essence a single user based converse bound. We can also derive a Fano type converse, but for the range of parameters we work with, it is worse than the presented one. The converse presented here illustrates the fact the $E_b/N_0$ requirements are necessarily higher compared to the AWGN channel in \citep{polyanskiy2017perspective}.

\begin{figure}[t]
\begin{center}
\input{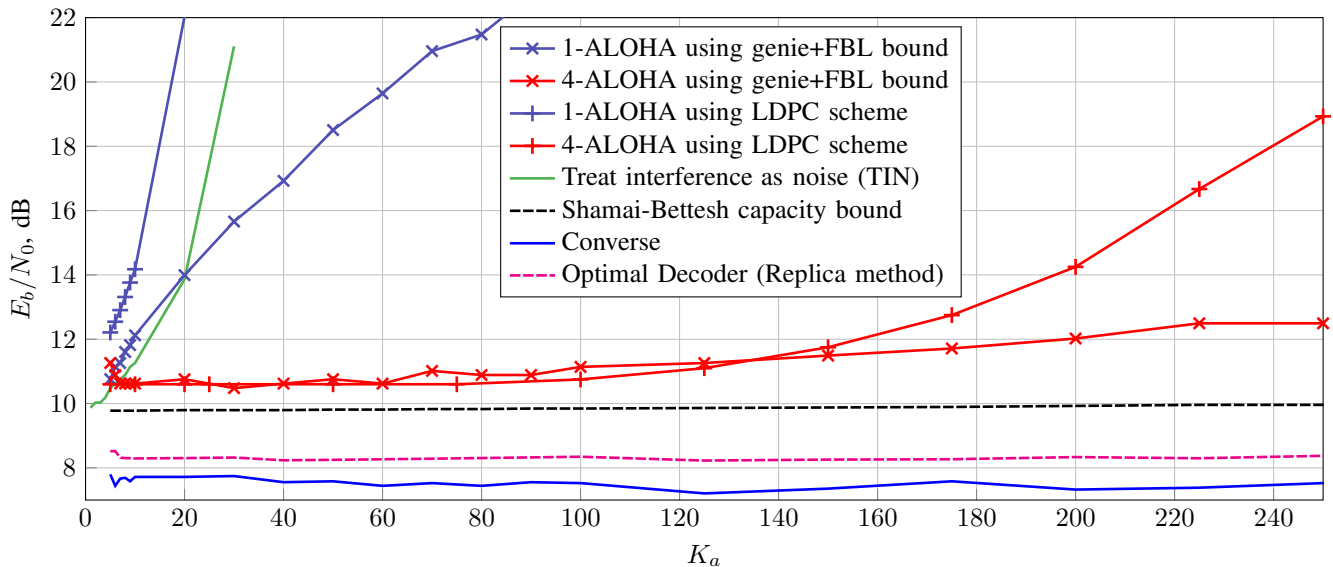}
\begin{tikzpicture}[define rgb/.code={\definecolor{mycolor}{RGB}{#1}},
                    rgb color/.style={define rgb={#1},mycolor}]
\begin{axis}[
    width=\FigureWidth, height=\FigureHeight,
    xlabel={$K_a$},
    ylabel={${E_b}/{N_0}$, dB},
    ylabel style={yshift=-0.6mm},
    xmin=0,    xmax=250,
    ymin=7, ymax=22,
    legend cell align={left},
    legend style={at={\LegendPositioning}, anchor=south west, nodes={scale=\LegendScale, transform shape}},
    axis line style={latex-latex},
    grid=both,
    grid style={line width=.1pt, draw=gray!20},
    major grid style={line width=.2pt,draw=gray!50},
    tick align=inside,
    tickpos=left
]
\addplot[
    rgb color={80, 80, 180},
    mark=x,
    mark size = \MarkerSize,
    line width=\LineWidth,
    ]
table[x=KA,y=EBNO]{tikz/ebno_ka_aloha/data/aloha1_fbl.dat};
\addplot[
    color=red,
    mark=x,
    mark size = \MarkerSize,
    line width=\LineWidth,
    ]
table[x=KA,y=EBNO]{tikz/ebno_ka_aloha/data/aloha4_fbl.dat};
\addplot[
    rgb color={80, 80, 180},
    mark=+,
    mark size = \MarkerSize,
    line width=\LineWidth,
    ]
table[x=KA,y=EBNO]{tikz/ebno_ka_aloha/data/aloha1_ldpc.dat};
\addplot[
    color=red,
    mark=+,
    mark size = \MarkerSize,
    line width=\LineWidth,
    ]
table[x=KA,y=EBNO]{tikz/ebno_ka_aloha/data/aloha4_ldpc.dat};
\addplot[
    rgb color={80, 180, 80},
    mark=none,
    line width=\LineWidth,
    ]
table[x=KA,y=EBNO]{tikz/ebno_ka_aloha/data/tin.dat};
\addplot[
    color=black,
    dash pattern={on 4pt off 1pt},
    mark=none,
    line width=\LineWidth,
    ]
table[x=KA,y=EBNO]{tikz/ebno_ka_aloha/data/shamai_bettesh.dat};
\addplot[
    color=blue,
    mark=none,
    line width=\LineWidth,
    ]
table[x=KA,y=EBNO]{tikz/ebno_ka_aloha/data/converse.dat};
\addplot[
    color=magenta,
    dash pattern={on 4pt off 1pt},
    mark=none,
    line width=\LineWidth,
    ]
table[x=KA,y=EBNO]{tikz/ebno_ka_aloha/data/optimal_decoder.dat};
\legend{
1-ALOHA using genie+FBL bound,
4-ALOHA using genie+FBL bound,
1-ALOHA using LDPC scheme,
4-ALOHA using LDPC scheme,
Treat interference as noise (TIN),
Shamai-Bettesh capacity bound,
Converse,
Optimal Decoder (Replica method)
}
\end{axis}
\end{tikzpicture}
\ifhmode\ifnum\lastnodetype=11 \unskip\fi\fi
\caption {$K_a$ vs $E_b/N_0$ for $\epsilon\leq 0.1$, $n=30000$, $k=100$ bits. Dashed lines  correspond to asymptotic approximation obtained by taking $n\to \infty$ and are shown only for reference.}
\label{fig:1}
\end{center}
\end{figure}
\begin{figure}[ht]
\centering
\input{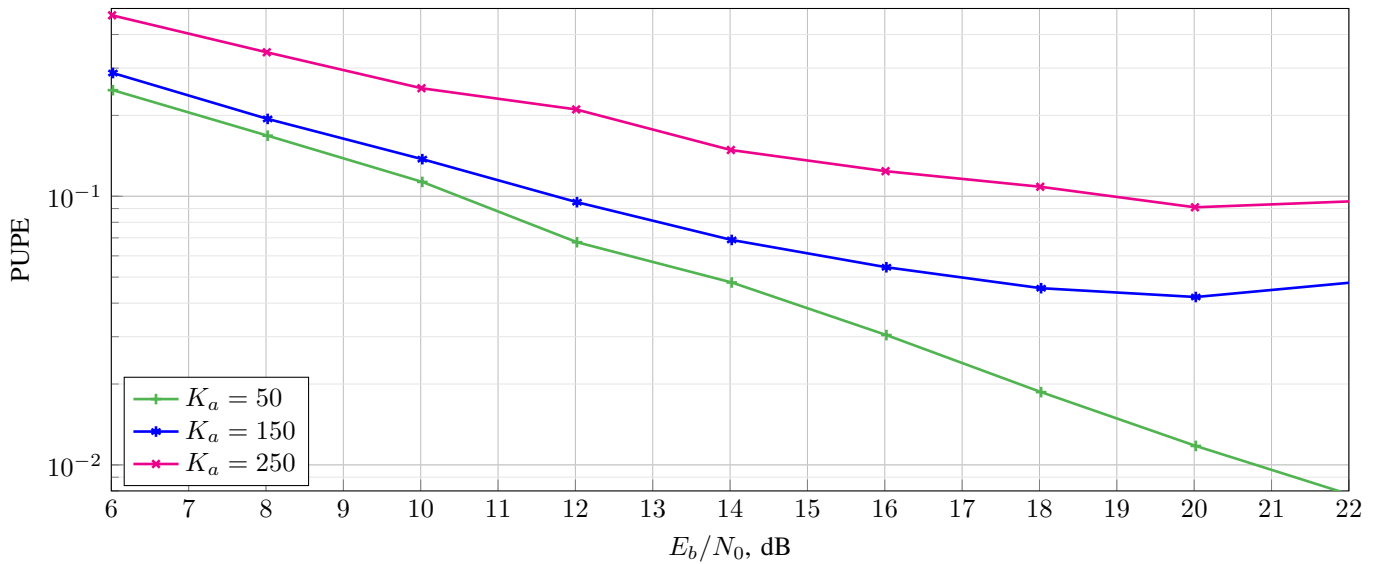}
\begin{tikzpicture}[define rgb/.code={\definecolor{mycolor}{RGB}{#1}},
                    rgb color/.style={define rgb={#1},mycolor}]
\begin{axis}[
    width=\FigureWidth, height=\FigureHeight,
    ymode=log,
    xlabel={${E_b}/{N_0}$, dB},
    ylabel={PUPE},
    xmin=6,    xmax=22,
    ymin=8e-3, ymax=0.5,
    legend cell align={left},
    legend style={at={(0.01,0.01)}, anchor=south west},
    axis line style={latex-latex},
    grid=both,
    grid style={line width=.1pt, draw=gray!20},
    major grid style={line width=.2pt,draw=gray!50},
    tick align=inside,
    tickpos=left
]
\addplot[
    rgb color={80, 180, 80},
    mark=+,
    mark size = \MarkerSize,
    line width=\LineWidth,
    ]
table[x=EBNO,y=PUPE]{tikz/ebno_pupe/data/ka_50.dat};

\addplot[
    color=blue,
    mark=asterisk,
    mark size = \MarkerSize,
    line width=\LineWidth,
    ]
table[x=EBNO,y=PUPE]{tikz/ebno_pupe/data/ka_150.dat};
\addplot[
    color=magenta,
    mark=x,
    mark size = \MarkerSize,
    line width=\LineWidth,
    ]
table[x=EBNO,y=PUPE]{tikz/ebno_pupe/data/ka_250.dat};
\legend{$K_a=50$, $K_a=150$, $K_a=250$}
\end{axis}
\end{tikzpicture}
\ifhmode\ifnum\lastnodetype=11 \unskip\fi\fi
\caption {$E_b/N_0$ vs $PUPE$ for  $n=30000$, $k=100$ bits}
\label{fig:2}
\end{figure}

\section{Asymptotics of Random-Access}\label{sec:asymp}

In \citep{polyanskiy2017perspective} the authors evaluated a random coding bound for AWGN RAC with
$n=30000$ and $K_a = 1,...,300$. The most interesting observation was that the bound on energy-per-bit 
was essentially constant up until about $K_a=150$ and only then started to increase with $K_a$. To explain this ''phase
transition'' behavior a particular asymptotics was postulated in \citep{ZPT-isit19}, which predicts the phase transition
at roughly the same value of $K_a=150$. It turned out that at low $K_a$ the performance was essentially limited by the
minimal energy required for a single user to send $k$ bits over a fixed (but effectively infinite) blocklength. For
larger number of $K_a$ the performance is limited by the multi-user requirement: the total number of $K_a \times k$ bits
should not exceed the combined mutual information of $n \log (1 + P K_a)$. 

In the present paper we adopt the very same asymptotics of~\citep{polyanskiy2017perspective,ZPT-isit19}. Again,
we stress that the only ultimately relevant question is the one at finite blocklength. The asymptotic analysis here is
only to get some insight into the possible regimes. 

Specifically, we consider scaling of $n\to\infty$ witht $K_a$, the number of active users, scaling linearly with
blocklength (similar to the \textit{many-access} regime
\citep{chen2017capacity,kowshik2019fundamental,polyanskiy2017perspective}) i.e., $K_a=\mu n$. At the same time, the size
of the common codebook is also scaling linearly: $M=M_1 K_a$. Since we operate in the same-codebook scenario this means
that the common codebook size scales linearly with number of users: $M=M_1 K_a $. We think of $M_1$ as the effective
payload per user. We also modify the random-access model slightly by requiring that the messages of active users
$\left\{W_1,...,W_{K_a}\right\}$ are sampled uniformly from $\binom{[M]}{K_a}$ i.e, user messages are sampled
\textit{uniformly without replacement} from $[M]$. (In reality, the user messages are distributed iid $\mathrm{Unif}[M]$
which leads to around $\frac{\binom{K_a}{2}}{M}$ collisions but for finite length scenarios with $M_1=2^{100}$, this is
essentially zero, hence we may ignore collisions in our asymptotic setup and simplify the analysis.)
If $P$ denotes the power (per symbol) of each user, then the energy-per-bit $E_b/N_0$ is defined by
\begin{equation}
\label{eq:ebn0_1}
E_b/N_0=\frac{n P}{\log M_1}.
\end{equation}

Hence, for finite $E_b/N_0$, we need the total sum-power $\PT=K_a P$ to be constant. Therefore, the asymptotic energy-per-bit, denoted by $\mathcal{E}$ is given by
\begin{equation}
\label{eq:ebn0_2}
\mathcal{E}=\frac{\PT}{\mu \log M_1}.
\end{equation}
We note that $E_b/N_0$ is defined this way for the reason that $\log {M\choose K_a} \approx K_a \log M$ for relevant finite-length values.

Lastly, the error metric is PUPE. We are interested in the trade-off of minimum $\mathcal{E}$ required to achieve a target PUPE with the user density $\mu$ as $n\to\infty$. 

This setup is equivalent to the support recovery in compressed sensing considered in \citep{reeves2012sampling,reeves2013approximate}. Here, we provide a comparison of the fundamental trade-off of energy-per-bit with user density, for given PUPE and $\rho$, between our analysis of the projection decoder, the ML decoder in \citep{reeves2012sampling}, the optimal decoder based on the true posteriors (see \citep[Theorem 8]{reeves2012sampling} for instance, this assumes replica symmetry to hold) and finally a converse.

To formally state our results we modify the definition of $(M,n,\epsilon)$ code for the $K_a$ user channel $P_{Y^n|X^n}$ given in \eqref{eq:sys1} as follows.
\begin{define}
\label{def2}
An $(M,n,\epsilon)$ random-access code for the $K_a$ user MAC $P_{Y^n|X^n}$ is a pair of (possibly randomized) maps $f:[M]\to \mathcal{X}^n$ (the encoder) and $g:\mathcal{Y}^n\to\binom{[M]}{K_a}$ such that if $W_1,...,W_{K_a}$ are sampled uniformly without replacement from $[M]$ and $X_j=f(W_j)$ then the average (per-user) probability of error satisfies
\begin{equation}
\label{eq:def2}
P_{e}=\frac{1}{K_a}\sum_{j=1}^{K_a}\Pb{W_j\notin g(Y^n)}\leq \epsilon
\end{equation}
where $Y^n$ is the channel output.
\end{define}
 
Define $(n,M,\epsilon,\mathcal{E},K_a)$--code as an $(M,n,\epsilon)$ random access code (from definition \ref{def2}) for the $K_a$--MAC with codebook $\mathcal{C}$ such that $\norm{c}^2\leq n P=\mathcal{E}\log M_1,\forall c\in\mathcal{C}$. Then we can define the following fundamental limit
\begin{equation}
\label{eq:ebn0_fund}
\mathcal{E}^*(M_1,\mu,\epsilon)=\limsup_{n\to\infty}\inf \left\{\mathcal{E}:\exists\left(n,M=K_a M_1,\epsilon,\mathcal{E},K_a=\mu n\right)-\mathrm{code}\right\}.
\end{equation}

In appendix \ref{app:asymp} we sandwich the fundamental limit between an achievability and a converse bound as follows:
\begin{equation}
\label{eq:ebn0_bound}
\mathcal{E}_{conv}\leq \mathcal{E}^*\leq \mathcal{E}_{ach}\,.
\end{equation}
For particular, quite cumbersome, expressions please refer to Appendix \ref{app:asymp}.

These bounds are plotted in figures \ref{fig:asymp1} and \ref{fig:asymp2} for two different values of PUPE.
The main achievability bound is from theorem \ref{th:asymp1} and is based on the analysis of projection decoding
described in appendix \ref{app:achiev}. A different analysis of this decoder was performed in \citep{reeves2012sampling}
and the result is plotted as well. We have also plotted predicted performance of the PUPE-optimal decoder for the iid codebook
which is obtained via a non-rigorous (but highly likely to be correct) replica-method from statistical physics; see appendix
\ref{sec:opt_dec} and \citep{reeves2012sampling}).

The converse bound plotted is based on Fano inequality and the single-user converse for AWGN channel from
\citep{polyanskiy2011minimum}. The details are in appendix \ref{sec:conv_asymp}. A tighter converse (see theorem
\ref{Th:conv2}) bound can be obtained if we assume that the codebook consists of iid entries of the form $\frac{C}{K_a}$
where is $C$ is of zero mean and finite variance. This follows from \citep[Theorem 37]{reeves2013approximate}. This
bound, although only applicable to a special class of codes (iid codebooks), improves our converse bound by taking into
account the penalty incurred due to absence of knowledge of the channel state information at the decoder (resulting in
a need to spend some of the information on estimating the fading coefficients).

\begin{figure}[t]
\centering
\input{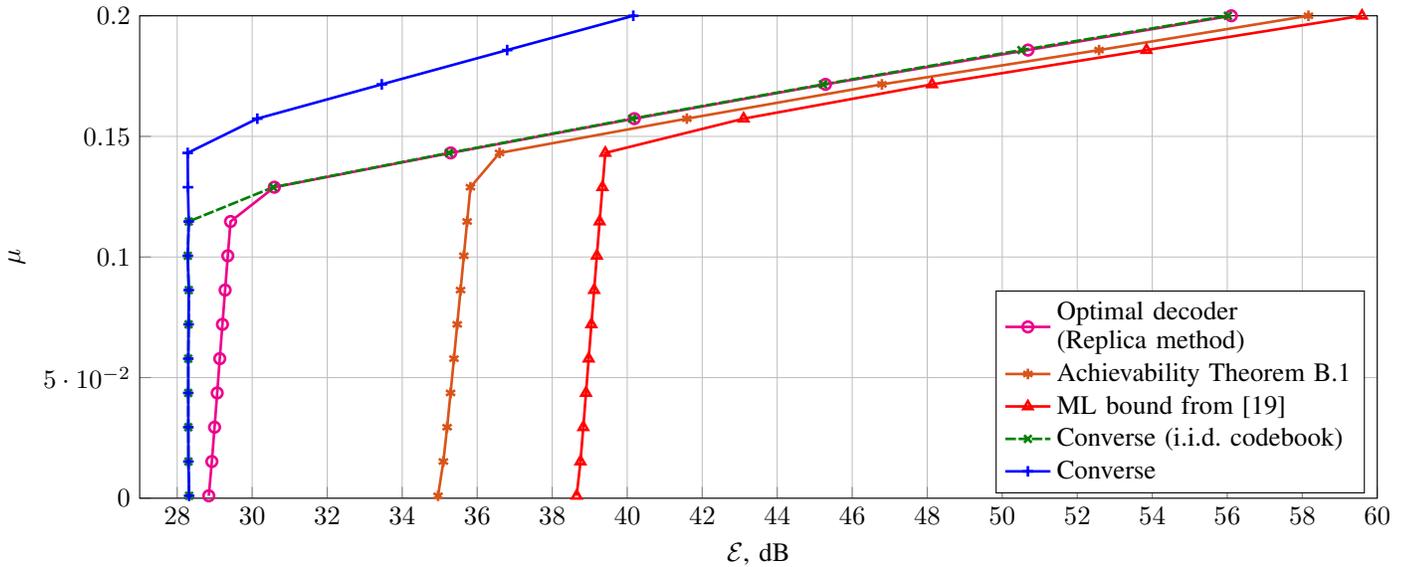}
\begin{tikzpicture}[define rgb/.code={\definecolor{mycolor}{RGB}{#1}},
                    rgb color/.style={define rgb={#1},mycolor}]

\def\DataPath{tikz/TIT_ebn0_mu/data}
\begin{axis}[
    width=\FigureWidth, height=\FigureHeight,
    xlabel={$\mathcal{E}$, dB},
    ylabel={$\mu$},
    xmin=\XMin,    xmax=\XMax,
    ymin=0, ymax=\YMax,
    legend cell align={left},
    legend style={at={(0.99,0.01)}, anchor=south east, nodes={scale=\LegendScale, transform shape}},
    axis line style={latex-latex},
    grid=both,
    grid style={line width=.1pt, draw=gray!20},
    major grid style={line width=.2pt,draw=gray!50},
    tick align=inside,
    tickpos=left
]

\addplot[
    color=magenta,
    mark=o,
    line width=\LineWidth,
    ]
table[x=EPS,y=MU]{\DataPath/pe\DataSuffix_optimal.dat};

\addplot[
    rgb color={217, 84, 26},
    mark=asterisk,
    mark size = \MarkerSize,
    line width=\LineWidth,
    ]
table[x=EPS,y=MU]{\DataPath/pe\DataSuffix_achievability.dat};

\addplot[
    color=red,
    mark=triangle,
    mark size = \MarkerSize,
    line width=\LineWidth,
    ]
table[x=EPS,y=MU]{\DataPath/pe\DataSuffix_ml.dat};

\addplot[
    rgb color={0, 128, 0},
    mark=x,
    dash pattern={on 4pt off 1pt},
    mark size = \MarkerSize,
    line width=\LineWidth,
    ]
table[x=EPS,y=MU]{\DataPath/pe\DataSuffix_converse_iid.dat};

\addplot[
    color=blue,
    mark=+,
    mark size = \MarkerSize,
    line width=\LineWidth,
    ]
table[x=EPS,y=MU]{\DataPath/pe\DataSuffix_converse.dat};

\addlegendentry[align=left]{Optimal decoder \\ (Replica method)}
\addlegendentry[align=left]{Achievability Theorem \ref{th:asymp1}}
\addlegendentry[align=left]{ML bound from \citep{reeves2012sampling}}
\addlegendentry[align=left]{Converse (i.i.d. codebook)}
\addlegendentry[align=left]{Converse}

\end{axis}
\end{tikzpicture}
\ifhmode\ifnum\lastnodetype=11 \unskip\fi\fi
\caption {$\mu$ vs $\mathcal{E}$ for $\epsilon\leq 10^{-3}$, $M_1=2^{100}$ }
\label{fig:asymp1}
\end{figure}

\begin{figure}[t]
\centering
\input{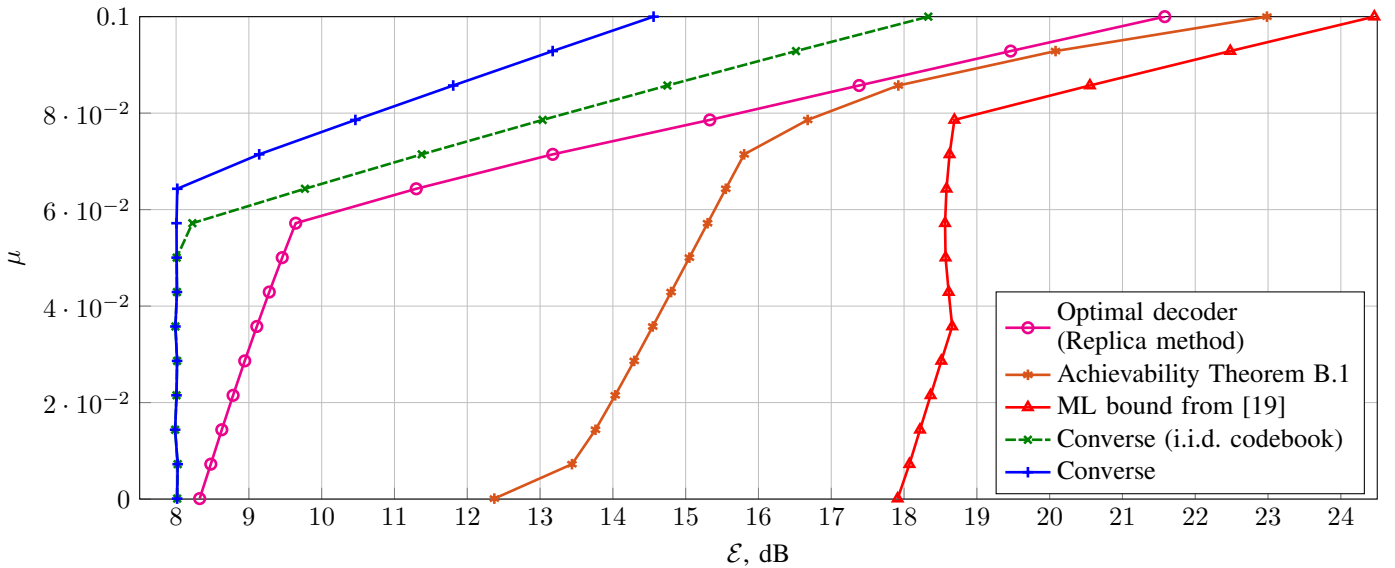}
\begin{tikzpicture}[define rgb/.code={\definecolor{mycolor}{RGB}{#1}},
                    rgb color/.style={define rgb={#1},mycolor}]

\def\DataPath{tikz/TIT_ebn0_mu/data}
\begin{axis}[
    width=\FigureWidth, height=\FigureHeight,
    xlabel={$\mathcal{E}$, dB},
    ylabel={$\mu$},
    xmin=\XMin,    xmax=\XMax,
    ymin=0, ymax=\YMax,
    legend cell align={left},
    legend style={at={(0.99,0.01)}, anchor=south east, nodes={scale=\LegendScale, transform shape}},
    axis line style={latex-latex},
    grid=both,
    grid style={line width=.1pt, draw=gray!20},
    major grid style={line width=.2pt,draw=gray!50},
    tick align=inside,
    tickpos=left
]

\addplot[
    color=magenta,
    mark=o,
    line width=\LineWidth,
    ]
table[x=EPS,y=MU]{\DataPath/pe\DataSuffix_optimal.dat};

\addplot[
    rgb color={217, 84, 26},
    mark=asterisk,
    mark size = \MarkerSize,
    line width=\LineWidth,
    ]
table[x=EPS,y=MU]{\DataPath/pe\DataSuffix_achievability.dat};

\addplot[
    color=red,
    mark=triangle,
    mark size = \MarkerSize,
    line width=\LineWidth,
    ]
table[x=EPS,y=MU]{\DataPath/pe\DataSuffix_ml.dat};

\addplot[
    rgb color={0, 128, 0},
    mark=x,
    dash pattern={on 4pt off 1pt},
    mark size = \MarkerSize,
    line width=\LineWidth,
    ]
table[x=EPS,y=MU]{\DataPath/pe\DataSuffix_converse_iid.dat};

\addplot[
    color=blue,
    mark=+,
    mark size = \MarkerSize,
    line width=\LineWidth,
    ]
table[x=EPS,y=MU]{\DataPath/pe\DataSuffix_converse.dat};

\addlegendentry[align=left]{Optimal decoder \\ (Replica method)}
\addlegendentry[align=left]{Achievability Theorem \ref{th:asymp1}}
\addlegendentry[align=left]{ML bound from \citep{reeves2012sampling}}
\addlegendentry[align=left]{Converse (i.i.d. codebook)}
\addlegendentry[align=left]{Converse}

\end{axis}
\end{tikzpicture}
\ifhmode\ifnum\lastnodetype=11 \unskip\fi\fi
\caption {$\mu$ vs $\mathcal{E}$ for $\epsilon\leq 0.1$, $M_1=2^{100}$ }
\label{fig:asymp2}
\end{figure}

\section{Conclusion and future work}
\label{sec:conclude}
In this work we considered random access for a quasi-static Rayleigh fading model. We developed low-complexity iterative
decoding scheme using LDPC codes to decode up to $T$--users in a slot, and using $T$--fold ALOHA on top of it gave us a
practical achievable scheme whose required $E_b/N_0$ vs $K_a$ trade-off is very close to that of a potential random
coding bound. In terms of future work, one of the most important things is to figure out how to relax the assumption on
the knowledge of the number of users in a slot in $T$--fold ALOHA to get a rigorous random coding achievability bound.
Another important factor is frame-synchronization
which we have assumed. Our rationale is that frame-synchronism can be achieved via regularly spaced beacons. However, to
reduce complexity even further it would be interesting to develop a beacon-free (and, hence, frame-asynchronous)
schemes. Finally, large gains in energy consumption can be attained via the use of MIMO, especially multiple receive
antennas. Quantifying these gains is yet another interesting direction.

\bibliographystyle{IEEEtran}
\bibliography{IEEEabrv,main}
\pagebreak

\appendices
\section{FBL achievability bounds}
\label{app:achiev}
In this section we state the random coding FBL achievability bounds for the model in \eqref{eq:sys1}.
But first, we discuss the encoding and decoding which we use to derive achievability. For encoding, we use random coding with \textbf{Gaussian codebook}: for each message a
$\cn(0,P'I_n)$ vector is independently  generated. That is $X_i\distas{iid}\cn(0,P'I_n)$ where $P'\leq P$. For a message $W_j$ of user $j$, if $\norm{X(W_j)}^2>nP$ then that user sends $0$.

\subsection{Projection decoding}
Inspired from \citep{yang2013quasi}, we use a projection based decoder. The idea is the following. Suppose there was no additive noise. Then the received vector will lie in the subspace spanned by the sent codewords no matter what the fading coefficients are. Fix an output list size $K_1$. The decoder outputs a list of $K_1$ codewords which form the subspace, such that projection of $Y$ onto this subspace is maximum.  Formally, let $C$ denote a set of vectors in $\mathbb{C}^n$. Denote $P_{C}$ as the orthogonal projection operator onto the subspace spanned by $C$. 

Let $\cd$ denote the common codebook. Then, upon receiving $Y$ from the channel, the decoder outputs $g(Y)$ given by
\begin{IEEEeqnarray}{C}
\label{eq:dec1}
g(Y)=\{f^{-1}(c): c\in\hat{C}\} \IEEEnonumber\\
\hat{C}=\arg\max_{C\subset \cd:|C|=K_1}\norm{P_{C}Y}^2
\end{IEEEeqnarray}
where $f$ is the encoding function.

The projection decoding is also called nearest-subspace decoding, and has been used in the compressed sensing literature \citep{wainwright2009information,reeves2012sampling,reeves2009note,reeves2008sparse}. One might prefer to view it as a kind of maximum likelihood (ML) decoding as well (and is called as such), since it is equivalent to 
\begin{\Ieee}{LLL}
\label{eq:dec2}
\hat{C}=\arg\max_{{C\subset \cd:|C|=K_1}}\max_{\{H_i:i\in C\}} P_{Y|X,H}\Ieeen\\ 
P_{Y|X,H}(y,\{x_i\},\{h_i\})=\frac{1}{\pi^n}e^{-\norm{y-\sum_i h_i x_i}^2}\Ieeen
\end{\Ieee}

 It can be shown that for the vanilla $K_a$--user quasi-static fading MAC (with different codebook and the usual joint probability of error) with no channel state information, projection decoding achieves $\epsilon$--capacity region $C_\epsilon$ of the MAC \citep{kowshik2019fundamental}.
\subsection{FBL Achievability bounds}
\begin{theorem}
\label{th:Ka1}
Fix $P'<P$. Let $K_1\leq K_2$. Then there exists an $(M,n,\epsilon)$ (with $\epsilon\geq \frac{K_2-K_1}{K_2}$) random access code for the $K_2$--MAC \eqref{eq:sys1} satisfying power constraint $P$ (see \eqref{eq:pow}) and 
\begin{equation}
\label{eq:Ka_1}
\epsilon\equiv P_e(M,n,K_2,P) \leq \frac{K_2-K_1}{K_2}+\frac{1}{K_2}\sum_{t=1}^{K_1}K_{1,t}p_t +p_0
\end{equation}
with
$$
p_0 =\frac{\binom{K_2}{2}}{M}+K_2\Pb{\frac{P'}{2}\sum_{i\in[2n]}W_i^2>nP}, \quad W_i\distas{iid}\mathcal{N}(0,1),
$$
and
{\allowdisplaybreaks
\begin{\Ieee}{LLL}
p_t &\leq &\inf_{\delta>0}\left(\binom{K_2}{K_{1,t}}e^{-(n-K_1)\delta} +\Pb{\bigcup_{\substack{S_0\subset [K_2]\\ |S_0|=K_{1,t}}} \left\{G(Y,S_0,c_{S_0},t)\geq V_{n,t} \right\}}\right) \label{eq:Ka2a}\Ieeen\\ 
\end{\Ieee}
where 
{\allowdisplaybreaks
\begin{\Ieee}{LLL}
G(Y,S_0,c_{S_0},t)=\frac{\norm{Y}^2-\max_{\substack{S_2\subset S_0\\|S_2|=t}}\norm{P_{c_{[S_2\cup  \left([K_2]\setminus S_0\right)]}}Y}^2}{\norm{Y}^2-\norm{P_{c_{[[K_2]\setminus S_0]}}Y}^2} \label{eq:Ka3aa}\\ 
\Ieeen\\
K_{1,t}=K_2-K_1+t \label{eq:Ka3a}\Ieeen\\
V_{n,t}=e^{-\tilde{V}_{n,t}}\label{eq:Ka3b}\Ieeen\\
\tilde{V}_{n,t}=\delta+R_1+s_t \label{eq:Ka3c}\Ieeen\\
s_t=\frac{\ln\binom{n'-1}{t-1}}{n-K_1} \label{eq:Ka3d}\Ieeen\\
R_1=\frac{\ln\binom{M-K_2}{t}}{n-K_1}\label{eq:Ka3e}\Ieeen\\
n'=n-K_1+t \label{eq:Ka3f}\Ieeen\\
\end{\Ieee}
}
and, $\cd=\{c_i:i \in [M]\}$ denotes the Gaussian codebook, $\{c_i:i\in [K_2]\}$ are the transmitted codewords, $c_{S}=\{c_i:i\in S\}$, $Y$ is the received vector.\\

Further, the right hand side of \eqref{eq:Ka2a} can be upper bounded as
\begin{\Ieee}{LLL}
p_t &\leq & \inf_{\substack{\delta>0\\ \delta_1>0 \\ 0<\delta_2<1}}\left[\binom{K_2}{K_{1,t}}\left(e^{-(n-K_1)\delta}+e^{-n'f_n(\delta_1)}+e^{-n'\frac{\delta_2^2}{2}}\right)+\right.\\
&&\left. \Pb{\min_{1\leq i\leq K_1-t+1} \frac{P'\sum_{j=i}^{i+t-1}|H_{(j)}|^2}{1+P'\sum_{j=i+t}^{K_{1,t}-1+i}|H_{(j)}|^2}\leq \frac{(1+\delta_1(1-V_{n,t}))V_{n,t}^{-1} -1}{1-\delta_2}}\right] \label{eq:Ka2}\Ieeen\\
\end{\Ieee}
where
\begin{\Ieee}{LLL}
f_n(\delta_1)= \delta_1+1+\frac{2V_{n,t}}{1-V_{n,t}}(1+\delta_1)
 -\sqrt{1+\frac{2V_{n,t}}{1-V_{n,t}}(1+\delta_1)}\sqrt{2\delta_1+1+\frac{2V_{n,t}}{1-V_{n,t}}(1+\delta_1)}\label{eq:Ka3g}\Ieeen
\end{\Ieee}

}
and $\{|H_{(j)}|^2:j\in[K_2]\}$ denotes the order statistics of fading powers (in decreasing order).
\end{theorem}

\begin{proof}
See appendix \ref{app:proof}.
\end{proof}

\begin{remark}
\label{remark:a}
We note that \eqref{eq:Ka2a} in the above theorem holds even in case of random coding with spherical codebook i.e., codewords distributed uniformly on the (complex) power shell with $p_0=\frac{\binom{K_2}{2}}{M}$. But \eqref{eq:Ka2} requires that the codebook is (complex) Gaussian.
\end{remark}

To compute \eqref{eq:Ka2a} we use Monte-Carlo simulation described in section \ref{sec:res} for small values of $K_2$. For moderated values of $K_2$, the computation of the probability of union of a combinatorially large number of events in \eqref{eq:Ka2a} is prohibitive. However, there is a computationally tractable bound (which is worse than \eqref{eq:Ka2a}) on $p_t$ that we present in appendix \ref{app:proof}. 

We make the following observation about $K_1$. When the number of active users $K_2$ is large, it is hard to decode the message of the user with least fading power, since its expectation is $\frac{1}{K_2}$. Consequently, this user becomes a bottleneck. So, intuitively, it makes sense to drop the users with very bad channel gains and decode the rest, and the definition of per-user probability of error makes this possible. Indeed, this was proposed in \citep{bettesh2000outages} where the joint multiuser detector drops a fraction of users with smallest gains such that the rate tuple of the remaining users is inside the (random) capacity region. So for each $K_2$, we can find the optimum $K_1$ which is the number of messages that are decoded in a frame. 

\section{Asymptotics of Random-Access}
\label{app:asymp}
In this section, we provide achievability and converse bounds on $\mathcal{E}^*$, defined in~\eqref{eq:ebn0_fund}.
\subsection{Achievability}

\begin{theorem}
\label{th:asymp1}
Consider the channel \eqref{eq:sys1} with $K_a=\mu n$ where $\mu< 1$. Fix $M_1>1$ and target PUPE $\epsilon$. Let $M=K_a M_1$ denote the size of the codebook and $\PT=K_aP$ be the total power. Let $h(p)=-p\ln p-(1-p)\ln(1-p),p\in[0,1]$. Fix $\nu \in (1-\epsilon,1]$. Let $\epsilon'=\epsilon-(1-\nu)$. Then if $\mathcal{E}>\mathcal{E}_{ach}=\sup_{\frac{\epsilon'}{\nu}< \theta\leq 1}\sup_{\xi\in[0,\nu(1-\theta)]}\frac{\PTX(\theta,\xi)}{\mu\log M_1}$, there exists a sequence of $\left(n,M=K_a M_1,\epsilon_n,\mathcal{E},K_a=\mu n\right)$ codes such that $\limsup_{n\to\infty} \epsilon_n\leq \epsilon$, where, for $\frac{\epsilon'}{\nu}< \theta\leq 1$ and $\xi\in[0,\nu(1-\theta)]$, 
{\allowdisplaybreaks
\begin{IEEEeqnarray*}{RLL}
\PTX(\theta,\xi)&=&\frac{\hat{f}(\theta,\xi)}{1-\hat{f}(\theta,\xi)\alpha\left(\xi+\nu\theta,\xi+1-\nu(1-\theta)\right)}\label{eq:sc_nc1a}\Ieeen\\
\hat{f}(\theta,\xi)&=&\frac{f(\theta)}{\alpha(\xi,\xi+\nu\theta)}\label{eq:sc_nc1b}\Ieeen\\
f(\theta)&=& \frac{\frac{1+\delta_1^*(1-V_\theta)}{V_\theta}-1}{1-\delta_2^*}\label{eq:sc_nc1c}\Ieeen\\
V_{\theta}&=&e^{-\tilde{V}_{\theta}}\label{eq:sc_nc1d}\Ieeen\\
\tilde{V}_{\theta} &=&\delta^*+\mu\frac{(M_1-1)}{1-\mu\nu}h\left(\frac{\theta\nu}{M_1-1}\right)+
\frac{1-\mu\nu(1-\theta)}{1-\mu\nu}h\left(\frac{\theta\mu\nu}{1-\mu\nu(1-\theta)}\right) \label{eq:sc_nc1e}\Ieeen\\
\delta^* &=&\frac{\mu h(1-\nu(1-\theta))}{1-\mu\nu}\label{eq:sc_nc1f}\Ieeen\\
c_\theta &=&\frac{2V_{\theta}}{1-V_{\theta}}\label{eq:sc_nc1g}\Ieeen\\
q_\theta &=&\frac{\mu h(1-\nu(1-\theta))}{1-\mu\nu(1-\theta)}\label{eq:sc_nc1h}\Ieeen\\
\delta_1^*&=& q_\theta(1+c_\theta)+
\sqrt{q_\theta^2 (c_\theta^2+2c_\theta)+2q_\theta(1+c_\theta)}\label{eq:sc_nc1i}\Ieeen\\
\delta_2^* &=& \inf \left\{x: 0<x<1, -\ln(1-x)-x>\frac{\mu h(1-\nu(1-\theta))}{1-\mu\nu(1-\theta)}\right\}\label{eq:sc_nc1j}\Ieeen\\
\alpha(a,b)&=& a\ln(a)-b\ln(b)+b-a. \label{eq:sc_nc1k}\Ieeen
\end{IEEEeqnarray*}

Hence $\mathcal{E}^*\leq \mathcal{E}_{ach}$.
}
\end{theorem}

The proof of the above theorem follows from \eqref{eq:Ka2} (theorem \ref{th:Ka1}) and ideas very similar to \citep[Theorem IV.1]{kowshik2019fundamental}. We omit the details.

\subsection{Optimal decoder}
\label{sec:opt_dec}
In this section we briefly describe the optimal decoder and its performance assuming replica symmetry. More details can be found in \citep{reeves2012sampling}. 

Let the codebook be $\mathcal{C}$. The optimal decoder for PUPE is the one which computes, for $c\in\mathcal{C}$, the posteriors $P_{c|Y^n}$ which is the probability, conditional on received vector $Y^n$, that $c$ is the list of transmitted codewords. Then, it outputs the list of codewords corresponding to top $K_a$ posteriors. Further, the system model is slightly modified in that each message is transmitted with probability $p=K_a/M=1/M_1$. In the limiting case, assuming replica symmetry, the posteriors converge to the posterior $\Pb{X\neq 0 |Y}$ of a scalar channel $Y=X+\sigma Z$ where $Z\distas{}\cn(0,1)$, $X$ is $\cn(0,1)$ with probability $p$ and $0$ with probability $1-p$ and is independent of $Z$. The value of $\sigma$ is given by (see \citep[Theorem 8]{reeves2012sampling}, but modified here for complex case)
\begin{\Ieee}{LLL}
\label{eq:sc_replica_1}
\sigma^2=\arg\min_{\tau>0}\left\{\frac{1}{\mu M_1}\log \tau M_1+\log(e)\frac{1}{\tau M_1 \PT}+I(X;X+\sqrt{\tau}Z)\right\}.\Ieeen
\end{\Ieee}
 The PUPE converges to $\Pb{\Pb{X\neq 0|Y}<T | X\neq 0}$ where $T$ satisfies $\Pb{\Pb{X\neq 0|Y}>T}=p$. Hence, we can find the minimum $\PT$ such that this PUPE of the scalar channel is at most $\epsilon$, and this gives another achievability bound (assuming replica symmetry) on $\mathcal{E}^*$.

\subsection{Converse}
\label{sec:conv_asymp}
We present a converse for $\mathcal{E}^*$ based on Fano inequality and using the results from \citep{tan2014fixed,polyanskiy2011minimum}

\begin{theorem}
Let $M=K_a M_1$ be the codebook size. Given $\epsilon\leq 1-\frac{K_a}{M}$ and $\mu$ such that $M_1>2$ then $\mathcal{E}^*(M_1,\mu,\epsilon)>\mathcal{E}_{conv}$ where $\mathcal{E}_{conv}=\max\{\mathcal{E}_{conv,1},\mathcal{E}_{conv,2}\}$ satisfies the following two bounds
\begin{enumerate}
\item \begin{\Ieee}{LLL}
\label{eq:sc_conv_2}
\mathcal{E}_{conv,1}= \inf \frac{\PT}{\mu\log M_1}\Ieeen
\end{\Ieee}
where infimum is taken over all $\PT>0$ that satisfies 
\begin{\Ieee}{LLL}
\label{eq:sc_conv_3}
\mu\theta\log M_1 -\epsilon\mu \log\left( M_1-1\right)  -\mu h_2(\epsilon)\leq \log\left(1+\alpha(1-\theta,1)\PT\right),\forall\theta\in[0,1]\Ieeen\IEEEeqnarraynumspace
\end{\Ieee}
and $\alpha$ is defined in \eqref{eq:sc_nc1k}.
\\
\item \begin{\Ieee}{LLL}
\label{eq:sc_conv_4}
\mathcal{E}_{conv,2}= \inf \frac{\PT}{\mu\log M_1}\Ieeen
\end{\Ieee}
where infimum is taken over all $\PT>0$ that satisfies 
\begin{\Ieee}{LLL}
\label{eq:sc_conv_5}
\epsilon \geq 1-\Ex{ Q\left(Q^{-1}\left(\frac{1}{M_1}\right)-\sqrt{\frac{2\PT}{\mu}|H|^2}\right)}\Ieeen
\end{\Ieee}
where $\mathcal{Q}$ is the complementary CDF function of the standard normal distribution.
\end{enumerate}

\end{theorem}
\begin{proof}
The proof of \eqref{eq:sc_conv_2}, \eqref{eq:sc_conv_3} is based of Fano inequality and genie argument. Let $W=\left(W_1,...,W_{K_a}\right)$ where $W_{i}\distas{iid}\mathrm{Unif}[M]$ denote the transmitted messages of $K_a$ users. Let $\hat{W}$ be the decoded list of messages. Then $\epsilon= P_e=\frac{1}{K_a}\sum_{j\in[K_a]}\Pb{W_j\notin \hat{W}}$.

Suppose a genie $G$ reveals to the decoder a set $S_1\subset[K_a]$ of transmitted messages $W_{S_1}=\{W_i:i\in S_1\}$ along with corresponding fading coefficients $H_{S_1}=\{H_i:i\in S_1\}$. A converse bound in this case is a converse for the actual problem (when there is no Genie). Hence the equivalent channel at the receiver becomes 
\begin{equation}
\label{eq:sc_conv_6}
Y_G=\sum_{i\in S_2}H_i X_i+Z
\end{equation}
where $S_2=[K_a]\setminus S_1$ with the decoder outputting a list $L_G=L(Y_G,W_{S_1},H_{S_1})$ of messages of size at most $K_a$ and PUPE 
$$P_e^G=\frac{1}{K_a}\sum_{j\in [K_a]}\Pb{W_j\notin L_G}$$.

First note that the optimal decoder (for PUPE) outputs a list of size exactly $K_a$ since otherwise PUPE can be strictly reduced by extending the list to size $K_a$ by adding random messages. Further, it must contain $W_{S_1}$ because if there is $j\in S_1$ such that $W_j\notin L_G$ then replacing one non-transmitted message in $L_G$ by $W_j$ strictly decreases PUPE. Let $E_i=1[W_i\notin L_G]$ and $\epsilon_i^G=\Ex{E_i}$. Note that $\epsilon_i^G=0$ for $i\in S_1$. Now standard Fano type arguments give, for $i\in S_2$,
\begin{\Ieee}{LLL}
\label{eq:sc_conv_7}
I(W_i;L_G)\geq \log M -h_2(\epsilon_i^G)-\epsilon_i^G \log(M-K_a)-(1-\epsilon_i^G)\log K_a.\Ieeen
\end{\Ieee}

Since
$$
I(W_{S_2};L_G)\geq \sum_{i\in S_2}I(W_i;L_G),
$$
we have
\begin{\Ieee}{LLL}
\label{eq:sc_conv_8}
I(W_{S_2};L_G)\geq |S_2|\log M-\sum_{i\in S_2}h_2(\epsilon_i^G)-\log\left(\frac{M}{K_a}-1\right)\sum_{i\in S_2}\epsilon_i^G -|S_2|\log K_a.\Ieeen
\end{\Ieee}

Further,
$$
I(W_{S_2};L_G)\leq n\Ex{\log\left(1+\frac{\PT}{K_a}\sum_{i\in S_2}|H_{i}|^2\right)}.
$$
Let
$$
P_e(S_2)=\frac{1}{|S_2|}\sum_{i\in S_2}\epsilon_i^G.
$$
Then 
$$
\frac{|S_2|}{K_a}P_e(S_2)=P_e^G\leq P_e.
$$
Hence we have
\begin{\Ieee}{LLL}
\label{eq:sc_conv_9}
\frac{n}{K_a}\Ex{\log\left(1+\frac{\PT}{K_a}\sum_{i\in S_2}|H_{i}|^2\right)}&\geq & \frac{|S_2|}{K_a}\log M-\frac{1}{K_a}\sum_{i\in S_2}h_2(\epsilon_i^G)-\log\left(\frac{M}{K_a}-1\right)\frac{|S_2|}{K_a}P_e(S_2) -\frac{|S_2|}{K_a}\log K_a\Ieeen\\
&\geq & \frac{|S_2|}{K_a}\log \frac{M}{K_a}-h_2\left(P_e^G\right)-\log\left(\frac{M}{K_a}-1\right)P_e^G \Ieeen
\end{\Ieee}
where the second inequality follows from Jensen's inequality and the fact that $M_1=\frac{M}{K_a}>2$. Since $P_e^G\leq P_e\leq 1-\frac{K_a}{M}$, $h_2(P_e^G)+\log\left(\frac{M}{K_a}-1\right)P_e^G\leq h_2(P_e)+\log\left(\frac{M}{K_a}-1\right)P_e$. The above equation holds for all $S_2$. Taking limit, with $|S_2|=\theta K_a$ and using results on strong laws of order statistics \citep{van1980strong} (see proof of \citep[Theorem IV.6]{kowshik2019fundamental}) gives the first part of the theorem.

For the second part, we have the following converse for a single user AWGN MAC $Y=X+Z,X,Y\in \mathbb{R}^{\infty}, Z_i\distas{iid}\mathcal{N}(0,1)$. Define an $(E,M,\epsilon)$ code for this channel: codewords $(c_1,...,c_M)$ with $\norm{c_i}^2\leq E$ and a decoder such that probability of error is smaller than $\epsilon$. Then from \citep{polyanskiy2011minimum} we have that any $(E,M,\epsilon)$ code satisfies
\begin{equation}
\label{eq:sc_conv_10}
\frac{1}{M}\geq Q\left(\sqrt{2E}+Q^{-1}\left(1-\epsilon\right)\right).
\end{equation}

Now, if the decoder were to output a list of size at most $K_a$ in the above and the error is defined as the probability that the transmitted message is not in the output list, then from the proof of \eqref{eq:sc_conv_10} in \citep{polyanskiy2011minimum} and ideas of meta-converse for list decoding \citep{tan2014fixed}, it can be easily verified that the above equation modifies to

\begin{equation}
\label{eq:sc_conv_11}
\frac{K_a}{M}\geq Q\left(\sqrt{2E}+Q^{-1}\left(1-\epsilon\right)\right).
\end{equation}

Hence using the ideas in proof of theorem \ref{th:con_main1} to reduce the problem to single user case with list decoding and translating  \eqref{eq:sc_conv_11} to quasi-static case as in the proof of \citep[Theorem IV.6]{kowshik2019fundamental}, the result in \eqref{eq:sc_conv_5} follows.

\end{proof}

Tighter converse bounds can be obtained if further assumptions are made on the codebook. For example, if we assume that each codebook consists of iid entries of the form $\frac{C}{K_a}$ where $C$ is sampled from a distribution with zero mean and finite variance, then we have the following converse bound from \citep[Theorem 3]{reeves2013approximate} (see \citep[Remark 3]{reeves2013approximate} as well).

\begin{theorem}
\label{Th:conv2}
Let $\mu=K_a/n<1$ be the user density and $M= K_a M_1$ be the codebook size such that $M_1>2$, and let the common codebook be generated such that each code symbol iid of the form $\frac{C}{K_a}$ where $C$ is of zero mean and variance $\PT$. Then in order for the codebook to achieve PUPE $\epsilon$ with high probability, the energy-per-bit $\mathcal{E}$ should satisfy
\begin{\Ieee}{LLL}
\label{eq:sc_conv_new1}
\mathcal{E}\geq \inf\frac{\PT}{\mu\log M_1}\Ieeen
\end{\Ieee}
where infimum is taken over all $\PT>0$ that satisfies 
\begin{\Ieee}{LLL}
\label{eq:sc_conv_new2}
h_2 \left(\frac{1}{M_1}\right)-\frac{1}{M_1}h_2(\epsilon)-\left(1-\frac{1}{M_1}\right)h_2\left(\frac{\epsilon}{M_1-1}\right)\leq \left(\mathcal{V}\left(\frac{1}{\mu M_1},\PT\right)-\frac{1}{M_1}\mathcal{V}\left(\frac{1}{\mu},\PT\right)\right)\log e \Ieeen
\end{\Ieee}
where $\mathcal{V}$ is given by \citep{reeves2013approximate}
\begin{\Ieee}{LLL}
\mathcal{V}(r,\gamma)=r\ln\left(1+\gamma-\mathcal{F}(r,\gamma)\right)+\ln\left(1+r\gamma-\mathcal{F}(r,\gamma)\right)-\frac{\mathcal{F}(r,\gamma)}{\gamma}\label{eq:sc_conv_new3a}\Ieeen\\
\mathcal{F}(r,\gamma)=\frac{1}{4}\left(\sqrt{\gamma\left(\sqrt{r}+1\right)^2+1}-\sqrt{\gamma\left(\sqrt{r}-1\right)^2+1}\right)^2\label{eq:sc_conv_new3b}\Ieeen\\
\end{\Ieee}

\end{theorem}

\section{Proof of theorem \ref{th:Ka1}}
\label{app:proof}
In this section, we present the proof of theorem \ref{th:Ka1}. We remark that \eqref{eq:Ka6} and \eqref{eq:Ka13} prove \eqref{eq:Ka2a}.

\begin{proof}
Let the common (complex) Gaussian codebook $\mathcal{C}$ of size $M$ and power $P'<P$ be generated, that is, for each $j\in[M]$, generate $c_j\distas{iid}\cn(0,P'I_n)$. Let $W_j$ denote the random (in $[M]$) the message of user $j$. The transmitted channel input is given by $X_j=c_{W_j}1\left\{\norm{c_{W_j}}^2\leq nP\right\}$. Let $K_1\leq K_a$ be the number of messages in the received signal that are decoded. The decoder searches of all $K_1$ sized subsets of $[M]$. The decoder output $g_{D}(Y)\in\mathcal{C}$ is given by
\begin{IEEEeqnarray*}{C}
\label{eq:Ka3}
\hat{C}=\arg\max_{\substack{C\subset \cd \\|C|=K_1}}\norm{P_{\{c:c\in C\}}Y}^2\\
g_D (Y)=  \left\{f^{-1}(c):c\in\hat{C}\right\} \Ieeen\\
\end{IEEEeqnarray*}
where $f$ is the encoding function. The probability of error is given by
\begin{IEEEeqnarray}{LLL}
\label{eq:Ka4}
P_e=\frac{1}{K_2} \sum_{j=1}^{K_2}\Pb{W_j\notin g(Y),\mathrm{\ or\ }\exists i\neq j, W_j= W_i  }.
\end{IEEEeqnarray}

Note that $W_1,...,W_{K_2}$ are sampled independently with replacement from $[M]$. We perform a change of measure by sampling  $W_1,...,W_{K_2}$ from $[M]$ \textit{without} replacement, and also change the measure of transmitted message from $$X_j=c_{W_j}1\left\{\norm{c_{W_j}}^2\leq nP\right\}$$ to $X_j=c^j_{W_j}$. Since $P_e$ is the expectation of a non-negative random variable bounded by $1$, this measure change adds a total variation distance which can bounded by 
$$p_0=\frac{\binom{K_2}{2}}{M}+K_2\Pb{\frac{\chi_2(2n)}{2n}>\frac{P}{P'}}\to 0\quad \textit{as} \quad n\to \infty,$$
where $\chi_2(d)$ is the distribution of sum of squares of $d$ iid standard normal random variables (the chi-square distribution). This follows from the same reasoning used in the main theorem in \citep{polyanskiy2017perspective}.
Henceforth we only consider the new measure. Now, $P_e$ can be bounded as
\begin{equation}
\label{eq:Ka6}
P_e\leq \Ex{\frac{1}{K_2}\sum_{j=1}^{K_2}1[W_j\notin g(Y)]}+p_0 
\leq \frac{K_2-K_1}{K_2}+\frac{1}{K_2}\sum_{t=1}^{K_1}p_{1,t} K_{1,t}+p_0
\end{equation}
where $K_{1,t}$ is given by \eqref{eq:Ka3a} and $p_{1,t}=\Pb{\sum_{j=1}^{K_2}1[W_j\notin g(Y)]=K_{1,t}}$.

Let $F_t=\left\{\sum_{j=1}^{K_2}1[W_j\notin g(Y)] =K_{1,t}\right\}$. W.l.o.g, we will assume that the transmitted message list is $S=[K_2]$ and hence the corresponding codewords are $\{c_1,c_2,...,c_{K_2}\}$. Let $c_{[S_0]}\equiv \{c_i:i\in S_0\}$ and $H_{[S_0]}\equiv \{H_i:i\in S_0\}$, where $S_0\subset [K_2]$. Further, let $c_{[S_1][S_2]}=c_{[S_1\cup S_2]}$. Conditioning on $c_{[K_2]}, H_{[K_2]}$ and $Z$, we have~(\ref{eq:Ka7})

\begin{\Ieee}{LLL}
\label{eq:Ka7}
\Pb{F_t|c_{[K_2]},H_{[K_2]},Z}&\leq  &\Pb{\exists S_0\subset [K_2]: |S|=K_{1,t}, \exists S_1\subset [M]\setminus [K_2]: |S_1|=t : \right.\\
&&\left. \left.  \norm{\Pcc{S_1}{[K_2]\setminus S_0}Y}^2>\max_{\substack{S_2\subset S_0\\|S_2|=t}}\norm{\Pcc{S_2}{[K_2]\setminus S_0}Y}^2  \right | c_{[K_2]},H_{[K_2]},Z }\\
& \leq &\Pb{\bigcup_{\substack{S_0\subset [K_2]\\ |S_0|=K_{1,t}}}\bigcup_{\substack{S_1\subset [M]\setminus [K_2 ]\\|S_1|=t}} \left. F(S_0,S_2^*,S_1,t) \right|c_{[K_2]},H_{[K_2]},Z },\Ieeen
\end{\Ieee}
 
where
$$
F(S_0,S_2^*,S_1,t)=\left\{  \norm{\Pcc{S_1}{[K_2]\setminus S_0}Y}^2>\norm{\Pcc{S_2^*}{[K_2]\setminus S_0}Y}^2\right\},
$$
and $S_2^*\subset S_0$ is a possibly random (depending only on $H_{[K_2]}$) subset of size $t$, to be chosen later. Next we will bound $\Pb{F(S_0,S_2^*,S_1,t)|c_{[K_2]},H_{[K_2]},Z}$.\\

For the sake of brevity, let $A_0=c_{[S_2^*][[K_2]\setminus S_0]}$, $A_1=c_{[[K_2]\setminus S_0]}$ and $B_1=c_{[S_1]}$. We have the following claim which follows from \citep[Claim 1]{kowshik2019fundamental}.

\begin{claim}[{\citep{kowshik2019fundamental}}]
\label{claim:Ka1}
For any $S_1\subset [M]\setminus [K_2]$ with $|S_1|=t$, conditioned on $c_{[K_2]}$, $H_{[K_2]}$ and $Z$, the law of $\norm{\Pcc{S_1}{[K_2]\setminus S_0}Y}^2$ is same as the law of $\norm{P_{A_1}Y}^2+\norm{(I-P_{A_1})Y}^2 \emph{Beta}(t,n-K_1)$ where $\emph{Beta}(a,b)$ is a beta distributed random variable with parameters $a$ and $b$.
\end{claim}

Therefore we have,
\begin{\Ieee}{LLL}
\label{eq:Ka8}
\Pb{F(S_0,S_2^*,S_1,t)|c_{[K_2]},H_{[K_2]},Z}
= \Pb{Beta(n-K_1,t)<G_{S_0}|c_{[K_2]},H_{[K_2]},Z }=F_{\beta}\left(G_{S_0};n-K_1,t\right)\Ieeen
\end{\Ieee}
where 
\begin{\Ieee}{LLL}
\label{eq:Ka9}
G_{S_0}=\frac{\norm{Y}^2-\norm{P_{A_0}Y}^2}{\norm{Y}^2-\norm{P_{A_1}Y}^2}.\Ieeen
\end{\Ieee}
Since $t\geq 1$, we have $F_\beta \left(G_{S_0};n-K_1,t\right)\leq \binom{n'-1}{t-1}G_{S_0}^{n-K_1}$, where $n'$ is given by \eqref{eq:Ka3f}.

Let us denote $\bigcup_{\substack{S_0\subset[K_2]\\|S|=K_{1,t}}}$ as $\bigcup_{S_0,K_1}$; similarly for $\sum$ and $\bigcap$ for the ease of notation. Using the above claim, we get,

\begin{\Ieee}{LLL}
\label{eq:Ka10}
\Pb{F_t|c_{[K_2]},H_{[K_2]},Z}\leq \sum_{S_0,K_1}\binom{M-K_2}{t}\binom{n'-1}{t-1}G_{S_0}^{n-K_1}.\Ieeen
\end{\Ieee}

Therefore $p_{1,t}$ can be bounded as
{\allowdisplaybreaks
\begin{\Ieee}{LLL}
\label{eq:sc_Ka11}
p_{1,t} & = \Pb{ F_t} \\
&\leq \Ex{\min\left\{1,\sum_{S_0,K_1}\binom{M-K_2}{t}\binom{n'-1}{t-1}G_{S_0}^{n-K_1} \right\}}\\
&=  \Ex{\min\left\{1,\sum_{S_0,K_1} e^{(n-K_1)(s_t+R_1)} G_{S_0}^{n-K_1} \right\}}\Ieeen
\end{\Ieee}
}
where $s_t$ and $R_1$ are given by \eqref{eq:Ka3d} and \eqref{eq:Ka3e} respectively.

For $\delta>0$, define $V_{n,t}$ as in \eqref{eq:Ka3b}. Let $E_1$ be the event
\begin{\Ieee}{LLL}
\label{eq:Ka12}
E_1 &=& \bigcap_{S_0,K_1}\left\{-\ln G_{S_0} - s_t-R_1>\delta\right\}\\
&=& \bigcap_{S_0,K_1}\left\{G_{S_0}<V_{n,t}\right\}.\Ieeen
\end{\Ieee}

Let $p_{2,t}=\Pb{\bigcup_{S_0,K_1}\left\{G_{S_0}>V_{n,t}\right\}}$. Then

{\allowdisplaybreaks
\begin{\Ieee}{LLL}
\label{eq:Ka13}
p_{1,t} &\leq & \Ex{\min\left\{1,\sum_{S_0,K_1}e^{(n-K_1)(s_t+R_1)}G_{S_0}^{n-K_1}\right\}\left(1[E_1]+1[E_1^c]\right)}\\
&\leq &  \Ex{\sum_{S_0,K_1} e^{-(n-K_1)\delta}}+p_{2,t}\\
&=& \binom{K_2}{K_{1,t}} e^{-(n-K_1)\delta}+p_{2,t}.\IEEEyesnumber
\end{\Ieee}
}

\textbf{Note:} This proves \eqref{eq:Ka2a}. 

Let us bound $p_{2,t}$. Let $\hat{Z}=Z+\sum_{i\in S_0\setminus S^*_2}H_i c_i$. From \citep[Claim 2]{kowshik2019fundamental} we have

\begin{claim}[{\citep{kowshik2019fundamental}}]
\label{claim:Ka2} $p_{2,t}$ is bounded as
\begin{IEEEeqnarray*}{LLL}
\label{eq:Ka14}
 p_{2,t} &=&\Pb{\bigcup_{S_0,K_1}\left\{G_{S_0}>V_{n,t}\right\}}\\
& \leq & \Pb{\bigcup_{S_0,K_1}\left\{\norm{(1-V_{n,t})\PPA \hat{Z}-V_{n,t}\PPA \sum_{i\in S_2^*}H_i c_i}^2 \geq V_{n,t}\norm{\PPA \sum_{i\in S^*_2}H_i c_i}^2\right\}}.\IEEEyesnumber
\end{IEEEeqnarray*}
\end{claim}

Let $\chi'_2(\lambda,d)$ denote the non-central chi-squared distributed random variable with non-centrality $\lambda$ and degrees of freedom $d$. That is, if $W_i\distas{}\mathcal{N}(\mu_i,1),i\in[d]$ and $\lambda=\sum_{i\in[d]}\mu_i^2$, then $\chi'_2(\lambda,d)$ has the same distribution as that of $\sum_{i\in[d]}W_i^2$. We have the following claim from \citep[Claim 3]{kowshik2019fundamental}.

\begin{claim}[{\citep{kowshik2019fundamental}}]
\label{claim:Ka3}
Conditional on $H_{[K_2]}$ and $A_0$,
 \begin{IEEEeqnarray*}{LLL}
\label{eq:Ka15a}
\norm{\PPA\left(\hat{Z}-\frac{V_{n,t}}{1-V_{n,t}}\hcc\right)}^2 
\distas{}\phso\frac{1}{2}\chi'_2\left(2F,2 n'\right)\IEEEyesnumber
\end{IEEEeqnarray*}
where 
\begin{IEEEeqnarray}{LLL}
F=\frac{\norm{\frac{V_{n,t}}{1-V_{n,t}}\PPA\hcc}^2}{\phso}\label{eq:Ka15b}\\
\end{IEEEeqnarray}

Hence its conditional expectation is
\begin{equation}
\label{eq:eq:Ka15c}
\mu=n'+F.
\end{equation}
\end{claim}

Now let
\begin{IEEEeqnarray}{LLL}
T=\frac{1}{2}\chi'_2(2F,2n')-\mu\label{eq:Ka16}\\
U=\frac{V_{n,t}}{(1-V_{n,t})}\frac{\norm{\PPA\hcc}^2}{\phso}-n' \label{eq:Ka17}\\
U^1=\frac{1}{1-V_{n,t}}\left(V_{n,t}W_{S_0}-1\right)\label{eq:Ka18}\IEEEeqnarraynumspace
\end{IEEEeqnarray}
where $W_{S_0}= \left(1+\frac{\norm{\PPA\hcc}^2}{n'\phso}\right)$. Notice that $U=n'U^1$ and $F=\frac{V_{n,t}}{1-V_{n,t}}n'(1+U^1)$.

Then we have \eqref{eq:Ka19}.
 
\begin{\Ieee}{LLL}
\label{eq:Ka19}
\mathrm{RHS\ of\ \eqref{eq:Ka14}}
&=&\Pb{\bigcup_{S_0,K_1}\left\{\norm{\PPA \hat{Z}-\frac{V_{n,t}}{(1-V_{n,t})}\PPA \hcc}^2 -\mu \geq U\right\}} \\
&=&\Pb{ \bigcup_{S_0,K_1}\left\{T\geq U\right\} }. \Ieeen
\end{\Ieee}
 
Now, let $\delta_1>0$, and $E_2=\cap_{S_0,K_1}\left\{U^1>\delta_1\right\}$. Taking expectations over $E_1$ and its complement, we have
\begin{\Ieee}{LLL}
\label{eq:Ka20}
\Pb{ \bigcup_{S_0,K_1}\left\{T\geq U\right\} } &\leq &\sum_{S_0,K_1} \Pb{T>U, U^1>\delta_1}+\Pb{E_2^c}\\
&= &\sum_{S_0,K_1}\Ex{\Pb{\left.T>U\right| H_{[K_2]},A_0}1[U^1>\delta_1]}+\Pb{E_2^c}\Ieeen
\end{\Ieee}
which follows from the fact that $\{U^1> \delta_1\}\in \sigma(H_{[K_2]},A_0)$. To bound this term, we use the following concentration result from \citep[Lemma 8.1]{birge2001alternative}.

\begin{lemma}[\citep{birge2001alternative}]
\label{lem:chi2}
Let $\chi=\chi_2'(\lambda,d)$ be a non-central chi-squared distributed variable with $d$ degrees of freedom and non-centrality parameter $\lambda$. Then $\forall x>0$
\begin{equation}
\label{eq:Ka21}
\begin{split}
& \Pb{\chi-(d+\lambda)\geq 2\sqrt{(d+2\lambda)x}+2x}\leq e^{-x}\\
& \Pb{\chi-(d+\lambda)\leq  -2\sqrt{(d+2\lambda)x}} \leq e^{-x}
\end{split}
\end{equation}
\end{lemma}

Hence, for $x>0$, we have
\begin{equation}
\label{eq:Ka22a}
\Pb{\chi-(d+\lambda)\geq x}\leq e^{-\frac{1}{2}\left(x+d+2\lambda-\sqrt{d+2\lambda}\sqrt{2x+d+2\lambda}\right)}.
\end{equation}
and for $x<(d+\lambda)$, we have
\begin{equation}
\label{eq:Ka22b}
\Pb{\chi\leq x}\leq e^{-\frac{1}{4}\frac{(d+\lambda-x)^2}{d+2\lambda}}.
\end{equation}

Observe that, in \eqref{eq:Ka22a}, the exponent is always negative for $x>0$ and finite $\lambda$ due to AM-GM inequality. When $\lambda=0$, we can get a better bound for the lower tail in \eqref{eq:Ka22b} by using \citep[Lemma 25]{reeves2012sampling}.
\begin{lemma}[\citep{reeves2012sampling}]
Let $\chi=\chi_2(d)$ be a chi-squared distributed variable with $d$ degrees of freedom. Then $\forall x>1$
\begin{\Ieee}{LLL}
\label{eq:Ka22c}
\Pb{\chi\leq \frac{d}{x}}\leq e^{-\frac{d}{2}\left(\ln x +\frac{1}{x}-1\right)}\Ieeen
\end{\Ieee}
\end{lemma}

Therefore, from \eqref{eq:Ka14}, \eqref{eq:Ka19}, \eqref{eq:Ka20} and \eqref{eq:Ka22a},  we have

\begin{IEEEeqnarray*}{LLL}
\label{eq:Ka23}
p_{2,t}\leq \sum_{S_0,K_1}\Ex{e^{-n'f_n(U^1)}1[U^1>\delta_1]}+\Pb{\bigcup_{S_0,K_1}\left\{U^1\leq \delta_1\right\}}\Ieeen
\end{IEEEeqnarray*}
where $f_n$ is given by \eqref{eq:Ka3g}.

Next, from \citep[Claim 4]{kowshik2019fundamental} we have that for $0<V_{n,t}<1$ and $x>0$, $f_n(x)$ is a monotonically increasing function of $x$. From this, we obtain
\begin{IEEEeqnarray*}{LLL}
\label{eq:Ka24}
p_{2,t}\leq \sum_{S_0,K_1}e^{-n'f_n(\delta_1)}+p_{3,t}\Ieeen
\end{IEEEeqnarray*}
where $p_{3,t}=\Pb{E_2^c}$.

Note that
 \begin{IEEEeqnarray}{LLL}
\label{eq:Ka25}
p_{3,t}=\Pb{E_2^c}= \Pb{\bigcup_{S_0,K_1}\left\{V_{n,t}W_{S_0}-1\leq \delta_1(1-V_{n,t})\right\}}.
 \end{IEEEeqnarray}
 Conditional on $H_{[K_2]}$, $\norm{\PPA\hcc}^2\distas{}\frac{1}{2}P'\sum_{i\in S_2^*}|H_i|^2\chi_2^{S_2^*}(2n')$, where $\chi_2(2n')$ is a chi-squared distributed random variable with $2n'$ degrees of freedom (here the superscript $S_2^*$ denotes the fact that this random variable depends on the codewords corresponding to $S_2^*$). For $1>\delta_2>0$, consider the event $E_4=\bigcap_{S_0,K_1} \left\{\frac{\chi_2^{S_2^*}(2n')}{2n'}> 1-\delta_2\right\}$. Using \eqref{eq:Ka22c}, we can bound $p_{3,t}$ as 
 
 \begin{IEEEeqnarray*}{LLL}
 \label{eq:Ka26}
 p_{3,t}\leq \sum_t\binom{K_2}{K_{1,t}} e^{-n'(-\ln(1-\delta_2)-\delta_2)}+p_{4,t} \IEEEyesnumber
 \end{IEEEeqnarray*}
 where 
 \begin{IEEEeqnarray*}{LLL}
 \label{eq:Ka27}
 p_{4,t}= \Pb{E_4^c}=\Pb{\bigcup_{S_0,K_1}\left\{V_{n,t}\left(1+\frac{P'\sum_{i\in S_2^*}|H_i|^2(1-\delta_2)}{\phso}\right)\leq  1+\delta_1(1-V_{n,t})\right\}}.\IEEEyesnumber
\end{IEEEeqnarray*}

We make an important observation here. The union bound over $S_0$ is the minimum over $S_0$, and it can be seen that optimum $S_0$ i.e, the minimizer  should be contiguous amongst the indices arranged according the decreasing order of fading powers. Then the best upper bound is got by choosing $S_2^*$ to be correspond to the top $t$ fading powers in $S_0$. Hence, we get

\begin{\Ieee}{LLL}
\label{eq:Ka28}
p_4=& \Pb{\min_{1\leq i\leq K_1-t+1} \frac{P'\sum_{j=i}^{i+t-1}|H_{(j)}|^2}{1+P'\sum_{j=i+t}^{K_{1,t}-1+i}|H_{(j)}|^2}\leq \frac{(1+\delta_1(1-V_{n,t}))V_{n,t}^{-1} -1}{1-\delta_2}}\Ieeen
\end{\Ieee}

Finally, combining \eqref{eq:Ka6}, \eqref{eq:Ka13}, \eqref{eq:Ka24}, \eqref{eq:Ka26} and \eqref{eq:Ka28} , and optimizing over $\delta$, $\delta_1$ and $\delta_2$, we are done.\\

\end{proof}

\section{Results for blind slot decoding}\label{app:slot}
Here we present the numerical results for blind slot decoding. Let us fix the following parameters:
\begin{itemize}
    \item $[400, 100]$ LDPC code for $4$-user case, obtained by PEXIT method in \citep{10.1007/978-3-030-01168-0_15};
    \item $25$ outer iterations, $50$ inner (LDPC) iterations;
    \item $T = 4$, which means that we can decode at most $4$ users in a slot;
\end{itemize}

We present the curves for $2$, $3$ and $4$ users, recall, that $T = 4$ for all the cases. We compare these curves with the following ``ideal'' curves
\begin{itemize}
\item fading channel coefficients are unknown, number of users is known (i.e. $T$ is selected to be equal to the actual number of users);
\item fading channel coefficients are known, number of users is known (full CSI).
\end{itemize}
Frame error rate performance for listed above scenarios are presented on~\Fig{fig:K2}, \Fig{fig:K3} and \Fig{fig:K4} for $K=2$, $3$, $4$ respectively. We see, that the performance curves for our coding scheme coincide with ``ideal'' curves and achievability bound and very close (the loss is less, than $2$~dB) to the converse bound. So we conclude, that LDPC-based scheme is good for resolving collisions of small order. 

\newcommand{\blindfer}[1]{
\begin{figure}[ht]
\centering
\def\NumUsers{#1}
\input{tikz/fer_blind/single_column.tex}
\begin{tikzpicture}[define rgb/.code={\definecolor{mycolor}{RGB}{#1}},
                    rgb color/.style={define rgb={#1},mycolor}]
\begin{axis}[
    width=\FigureWidth, height=\FigureHeight,
    ymode=log,
    xlabel={${E_b}/{N_0}$, dB},
    ylabel={$P_e$},
    xmin=6,    xmax=22.1,
    ymin=6e-3, ymax=0.4,
    legend cell align={left},
    legend style={at={(0.99,0.99)}, anchor=north east},
    axis line style={latex-latex},
    grid=both,
    grid style={line width=.1pt, draw=gray!20},
    major grid style={line width=.2pt,draw=gray!50},
    tick align=inside,
    tickpos=left
]
\addplot[
    color=red,
    mark=diamond,
    mark size = \MarkerSize,
    line width=\LineWidth,
    ]
table[x=EBNO,y=FER]{tikz/fer_blind/data/k\NumUsers_blind.dat};

\ifthenelse{\equal{\NumUsers}{4}}
{}
{
\addplot[
    color=blue,
    mark=+,
    mark size = \MarkerSize,
    line width=\LineWidth,
    ]
table[x=EBNO,y=FER]{tikz/fer_blind/data/k\NumUsers_known_k.dat};
}
\addplot[
    rgb color={0, 128, 0},
    mark=square,
    mark size = \MarkerSize,
    line width=\LineWidth,
    ]
table[x=EBNO,y=FER]{tikz/fer_blind/data/k\NumUsers_known_kh.dat};
\addplot[
    rgb color={0, 115, 189},
    mark=o,
    mark size = \MarkerSize,
    line width=\LineWidth,
    ]
table[x=EBNO,y=FER]{tikz/fer_blind/data/k\NumUsers_fbl.dat};
\addplot[
    rgb color={217, 84, 26},
    mark=x,
    mark size = \MarkerSize,
    line width=\LineWidth,
    ]
table[x=EBNO,y=FER]{tikz/fer_blind/data/k\NumUsers_converse.dat};
\ifthenelse{\equal{\NumUsers}{4}}
{
\legend{
$K=\NumUsers$,
$K=\NumUsers${,} known $H$,
$K=\NumUsers${,} FBL achievability,
$K=\NumUsers${,} Converse,
}
}
{
\legend{
$K=\NumUsers${,} blind,
$K=\NumUsers${,} known $K$,
$K=\NumUsers${,} known $H$ and $K$,
$K=\NumUsers${,} FBL achievability,
$K=\NumUsers${,} Converse,
}
}
\end{axis}
\end{tikzpicture}
\caption{Simulation results for $K= \NumUsers$ users}
\label{fig:K\NumUsers}
\end{figure}
}

\begin{figure}[ht]
\centering
\def\NumUsers{2}
\input{tikz/fer_blind/single_column.tex}
\begin{tikzpicture}[define rgb/.code={\definecolor{mycolor}{RGB}{#1}},
                    rgb color/.style={define rgb={#1},mycolor}]
\begin{axis}[
    width=\FigureWidth, height=\FigureHeight,
    ymode=log,
    xlabel={${E_b}/{N_0}$, dB},
    ylabel={$P_e$},
    xmin=6,    xmax=22.1,
    ymin=6e-3, ymax=0.4,
    legend cell align={left},
    legend style={at={(0.99,0.99)}, anchor=north east},
    axis line style={latex-latex},
    grid=both,
    grid style={line width=.1pt, draw=gray!20},
    major grid style={line width=.2pt,draw=gray!50},
    tick align=inside,
    tickpos=left
]
\addplot[
    color=red,
    mark=diamond,
    mark size = \MarkerSize,
    line width=\LineWidth,
    ]
table[x=EBNO,y=FER]{tikz/fer_blind/data/k\NumUsers_blind.dat};

\ifthenelse{\equal{\NumUsers}{4}}
{}
{
\addplot[
    color=blue,
    mark=+,
    mark size = \MarkerSize,
    line width=\LineWidth,
    ]
table[x=EBNO,y=FER]{tikz/fer_blind/data/k\NumUsers_known_k.dat};
}
\addplot[
    rgb color={0, 128, 0},
    mark=square,
    mark size = \MarkerSize,
    line width=\LineWidth,
    ]
table[x=EBNO,y=FER]{tikz/fer_blind/data/k\NumUsers_known_kh.dat};
\addplot[
    rgb color={0, 115, 189},
    mark=o,
    mark size = \MarkerSize,
    line width=\LineWidth,
    ]
table[x=EBNO,y=FER]{tikz/fer_blind/data/k\NumUsers_fbl.dat};
\addplot[
    rgb color={217, 84, 26},
    mark=x,
    mark size = \MarkerSize,
    line width=\LineWidth,
    ]
table[x=EBNO,y=FER]{tikz/fer_blind/data/k\NumUsers_converse.dat};
\ifthenelse{\equal{\NumUsers}{4}}
{
\legend{
$K=\NumUsers$,
$K=\NumUsers${,} known $H$,
$K=\NumUsers${,} FBL achievability,
$K=\NumUsers${,} Converse,
}
}
{
\legend{
$K=\NumUsers${,} blind,
$K=\NumUsers${,} known $K$,
$K=\NumUsers${,} known $H$ and $K$,
$K=\NumUsers${,} FBL achievability,
$K=\NumUsers${,} Converse,
}
}
\end{axis}
\end{tikzpicture}
\caption{Simulation results for $K= \NumUsers$ users}
\label{fig:K\NumUsers}
\end{figure}

\begin{figure}[ht]
\centering
\def\NumUsers{3}
\input{tikz/fer_blind/single_column.tex}
\begin{tikzpicture}[define rgb/.code={\definecolor{mycolor}{RGB}{#1}},
                    rgb color/.style={define rgb={#1},mycolor}]
\begin{axis}[
    width=\FigureWidth, height=\FigureHeight,
    ymode=log,
    xlabel={${E_b}/{N_0}$, dB},
    ylabel={$P_e$},
    xmin=6,    xmax=22.1,
    ymin=6e-3, ymax=0.4,
    legend cell align={left},
    legend style={at={(0.99,0.99)}, anchor=north east},
    axis line style={latex-latex},
    grid=both,
    grid style={line width=.1pt, draw=gray!20},
    major grid style={line width=.2pt,draw=gray!50},
    tick align=inside,
    tickpos=left
]
\addplot[
    color=red,
    mark=diamond,
    mark size = \MarkerSize,
    line width=\LineWidth,
    ]
table[x=EBNO,y=FER]{tikz/fer_blind/data/k\NumUsers_blind.dat};

\ifthenelse{\equal{\NumUsers}{4}}
{}
{
\addplot[
    color=blue,
    mark=+,
    mark size = \MarkerSize,
    line width=\LineWidth,
    ]
table[x=EBNO,y=FER]{tikz/fer_blind/data/k\NumUsers_known_k.dat};
}
\addplot[
    rgb color={0, 128, 0},
    mark=square,
    mark size = \MarkerSize,
    line width=\LineWidth,
    ]
table[x=EBNO,y=FER]{tikz/fer_blind/data/k\NumUsers_known_kh.dat};
\addplot[
    rgb color={0, 115, 189},
    mark=o,
    mark size = \MarkerSize,
    line width=\LineWidth,
    ]
table[x=EBNO,y=FER]{tikz/fer_blind/data/k\NumUsers_fbl.dat};
\addplot[
    rgb color={217, 84, 26},
    mark=x,
    mark size = \MarkerSize,
    line width=\LineWidth,
    ]
table[x=EBNO,y=FER]{tikz/fer_blind/data/k\NumUsers_converse.dat};
\ifthenelse{\equal{\NumUsers}{4}}
{
\legend{
$K=\NumUsers$,
$K=\NumUsers${,} known $H$,
$K=\NumUsers${,} FBL achievability,
$K=\NumUsers${,} Converse,
}
}
{
\legend{
$K=\NumUsers${,} blind,
$K=\NumUsers${,} known $K$,
$K=\NumUsers${,} known $H$ and $K$,
$K=\NumUsers${,} FBL achievability,
$K=\NumUsers${,} Converse,
}
}
\end{axis}
\end{tikzpicture}
\caption{Simulation results for $K= \NumUsers$ users}
\label{fig:K\NumUsers}
\end{figure}

\begin{figure}[ht]
\centering
\def\NumUsers{4}
\input{tikz/fer_blind/single_column.tex}
\begin{tikzpicture}[define rgb/.code={\definecolor{mycolor}{RGB}{#1}},
                    rgb color/.style={define rgb={#1},mycolor}]
\begin{axis}[
    width=\FigureWidth, height=\FigureHeight,
    ymode=log,
    xlabel={${E_b}/{N_0}$, dB},
    ylabel={$P_e$},
    xmin=6,    xmax=22.1,
    ymin=6e-3, ymax=0.4,
    legend cell align={left},
    legend style={at={(0.99,0.99)}, anchor=north east},
    axis line style={latex-latex},
    grid=both,
    grid style={line width=.1pt, draw=gray!20},
    major grid style={line width=.2pt,draw=gray!50},
    tick align=inside,
    tickpos=left
]
\addplot[
    color=red,
    mark=diamond,
    mark size = \MarkerSize,
    line width=\LineWidth,
    ]
table[x=EBNO,y=FER]{tikz/fer_blind/data/k\NumUsers_blind.dat};

\ifthenelse{\equal{\NumUsers}{4}}
{}
{
\addplot[
    color=blue,
    mark=+,
    mark size = \MarkerSize,
    line width=\LineWidth,
    ]
table[x=EBNO,y=FER]{tikz/fer_blind/data/k\NumUsers_known_k.dat};
}
\addplot[
    rgb color={0, 128, 0},
    mark=square,
    mark size = \MarkerSize,
    line width=\LineWidth,
    ]
table[x=EBNO,y=FER]{tikz/fer_blind/data/k\NumUsers_known_kh.dat};
\addplot[
    rgb color={0, 115, 189},
    mark=o,
    mark size = \MarkerSize,
    line width=\LineWidth,
    ]
table[x=EBNO,y=FER]{tikz/fer_blind/data/k\NumUsers_fbl.dat};
\addplot[
    rgb color={217, 84, 26},
    mark=x,
    mark size = \MarkerSize,
    line width=\LineWidth,
    ]
table[x=EBNO,y=FER]{tikz/fer_blind/data/k\NumUsers_converse.dat};
\ifthenelse{\equal{\NumUsers}{4}}
{
\legend{
$K=\NumUsers$,
$K=\NumUsers${,} known $H$,
$K=\NumUsers${,} FBL achievability,
$K=\NumUsers${,} Converse,
}
}
{
\legend{
$K=\NumUsers${,} blind,
$K=\NumUsers${,} known $K$,
$K=\NumUsers${,} known $H$ and $K$,
$K=\NumUsers${,} FBL achievability,
$K=\NumUsers${,} Converse,
}
}
\end{axis}
\end{tikzpicture}
\caption{Simulation results for $K= \NumUsers$ users}
\label{fig:K\NumUsers}
\end{figure}

\section{Single-component GM performance}
\begin{figure}[ht]
\centering
\input{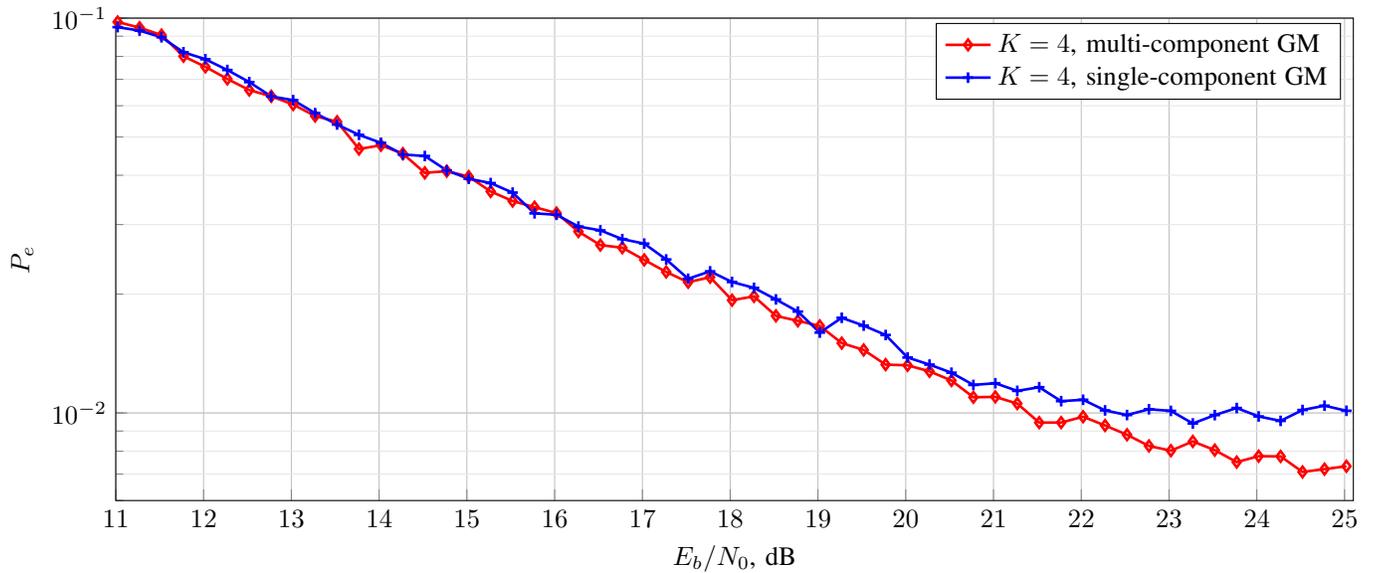}
\begin{tikzpicture}[define rgb/.code={\definecolor{mycolor}{RGB}{#1}},
                    rgb color/.style={define rgb={#1},mycolor}]
\begin{axis}[
    width=\FigureWidth, height=\FigureHeight,
    ymode=log,
    xlabel={${E_b}/{N_0}$, dB},
    ylabel={$P_e$},
    xmin=11,    xmax=25.1,
    ymin=6e-3, ymax=0.1,
    legend cell align={left},
    legend style={at={(0.99,0.99)}, anchor=north east},
    axis line style={latex-latex},
    grid=both,
    grid style={line width=.1pt, draw=gray!20},
    major grid style={line width=.2pt,draw=gray!50},
    tick align=inside,
    tickpos=left
]
\addplot[
    color=red,
    mark=diamond,
    mark size = \MarkerSize,
    line width=\LineWidth,
    ]
table[x=EBNO,y=FER]{tikz/hard_decision/data/gm_full.dat};

\addplot[
    color=blue,
    mark=+,
    mark size = \MarkerSize,
    line width=\LineWidth,
    ]
table[x=EBNO,y=FER]{tikz/hard_decision/data/gm_simple.dat};


\legend{
$K=4${,} multi-component GM,
$K=4${,} single-component GM,
}
\end{axis}
\end{tikzpicture}
\caption{Simulation results for $K=4$ users. Single component GM and multiple-component GM model (including component merge and prune model) frame error rate performance}
\label{fig:hard_decision}
\end{figure}
\begin{table}[h]
\caption{GM parameters for single and multiple component GM model\label{tab:sim_settings_gm}}
\begin{tabular}{|l|r|r|}
\hline
\multicolumn{1}{|c|}{Parameter} & Multi-component GM & Single-component GM \\
\hline
Gaussian mixture merge distance ($d_{min}$)              &    1      & $-$ \\
Gaussian mixture maximum component count ($\nu$)         &  500      & 1   \\
GM sample count to evaluate~\eqref{eq:decode_func_nodes} &   20      & 20  \\
Maximum cumulative weight to drop at prune (The components with the least weights are dropped before prune)  & $10^{-3}$ & $-$ \\
\hline
\end{tabular}
\end{table}
Let us consider how the GM configuration affects the overall decoding performance. The algorithm complexity highly depends on the maximum number of components $\nu$ allowed in the GM. Merge and prune procedures keep the maximum component count under some threshold. To address this issue, we have evaluated the frame error rate performance for $K=4$ users with the decoder having different settings. In the first setup, we have utilized multiple-component GM with merge and prune procedures (as before). The second setup assumes single-component GM with the merge procedure being disabled. Let's again consider the same $[400, 100]$ LDPC code as in the previous Appendix~\ref{app:slot}.

The second setup should be explained in more details. Recall to the four message types described in section~\ref{sec:ldpc}. Each GM at every message passing step consists of a single component with the highest probability and the sampling (required to evaluate~\eqref{eq:decode_func_nodes}) is performed from a single Gaussian distribution. As soon as the only GM component retains after every message type passing, there is no need to perform the merge procedure. It is worth to note that the merge procedure can only increase the covariance of the components to retain in the merge list. The most important change in the decoding algorithm (see section~\ref{sec:ldpc}) with single-component GM is equation~\eqref{eq:update_h_R}. The most probable symbol is considered in this case (because the second alternative for BPSK constellation point will be immediately dropped by prune procedure under the limit $\nu=1$). Detailed difference in the GM parameters is shown in Table~\ref{tab:sim_settings_gm}.

The frame error rate performance is shown on Figure~\ref{fig:hard_decision}. Let us explain all the curves in the figure.
\begin{itemize}
\item Red curve corresponds to our most complex decoder from Appendix~\ref{app:slot}, which utilizes Gaussian mixtures with a large number of components ($\nu = 500$, see Table~\ref{tab:sim_settings_gm}). The decoder performs merge and prune operations to guarantee that the number of components is less or equal than $\nu$. In this case each message is a pair of vectors $(\overline{\mu}, \overline{\sigma})$ -- means and variances, each vector is of length $\nu$. Real and imaginary components of fading coefficients estimates were represented by different mean and covariance vectors.
\item Blue curve corresponds to the case when $\nu = 1$. We still perform prune operation but do not perform merge operation. At each step, the most probable component is chosen. So, in this case, the message is a pair of scalar values $(\mu, \sigma)$ (again, real and imaginary parts are considered separately). The decoder has a surprisingly good performance.
\end{itemize}

One can see that the GM configuration affects the performance only at higher $E_b/N_0$. We see that the simpler the decoder the higher the error floor. For the blue and red curves we decided to perform simulation in $E_b/N_0$ range $[20, 25]$~dB to verify if error floor of the blue curve is higher. 

An important moment is that all the decoders do several independent decoding attempts as described above. As described in section~\ref{sec:blind_decoder}, multiple attempts are needed to guarantee that decoder will not fall in to local maximum of~\eqref{eq:argmax_functional}. Otherwise, the performance is bad. This can be the explanation of the fact that single-component GM works fine.

\end{document}